\documentclass[twocolumn]{IEEEtran}

\usepackage{cite}

\ifCLASSINFOpdf
  \usepackage[pdftex]{graphicx}
  \DeclareGraphicsExtensions{.pdf,.jpeg,.png}
\else
\fi
\usepackage[cmex10]{amsmath}
\usepackage{algorithm}
\usepackage{algpseudocode}

\usepackage{array}

\ifCLASSOPTIONcompsoc
  \usepackage[caption=false,font=normalsize,labelfont=sf,textfont=sf]{subfig}
\else
  \usepackage[caption=false,font=footnotesize]{subfig}
\fi
\usepackage{amsthm}
\newtheorem{theorem}{Theorem}
\newtheorem{lemma}{Lemma}
\newtheorem{proposition}{Proposition}

\theoremstyle{definition} \newtheorem{definition}{Definition}
\newtheorem{condition}{Condition}
\newtheorem{assumption}{Assumption}
\newtheorem{remark}{Remark}
\newtheorem{problem}{Problem}

\usepackage{color}
\usepackage[monochrome]{xcolor} 

\usepackage[normalem]{ulem}

\usepackage{amsfonts}
\usepackage{amssymb}

\newcommand{\col}{\color{red}}
\newcommand{\colb}{\color{blue}}
\newcommand{\colg}{\color{green!45!black}} 

\begin{document}
%
\title{Anonymous Hedonic Game for Task Allocation in a Large-Scale Multiple Agent System}
%
%
%

\author{Inmo~Jang,~
        Hyo-Sang~Shin,~
        and~Antonios~Tsourdos
\thanks{Inmo Jang, Hyo-Sang Shin, and Antonios Tsourdos are with Centre for Autonomous and Cyber-Physical Systems, 
Cranfield University, MK43 0AL, United Kingdom
(e-mail: inmo.jang@cranfield.ac.uk; h.shin@cranfield.ac.uk; a.tsourdos@cranfield.ac.uk). }
}

\maketitle

\begin{abstract} 
This paper proposes a novel game-theoretical autonomous decision-making framework to address a task allocation problem for a swarm of multiple agents. We consider cooperation of self-interested agents, and show that our proposed decentralized algorithm guarantees convergence of agents with \emph{social inhibition} to a Nash stable partition (i.e., social agreement) within polynomial time. The algorithm is simple and executable based on local interactions with {\col neighbor} agents under a strongly-connected communication network and even in asynchronous environments. We analytically present a mathematical formulation for computing the lower bound of suboptimality of the solution, and additionally show that 50\% of suboptimality can be {\colg at least} guaranteed if social utilities are non-decreasing functions with respect to the number of co-working agents. The results of numerical experiments confirm that the proposed framework is scalable, fast adaptable against dynamical environments, and robust even in a realistic situation. 
\end{abstract}

\begin{IEEEkeywords}
Distributed robot systems, 
Networked robots, 
Task allocation,
Game theory,
Self-organising systems
\end{IEEEkeywords}

%
\IEEEpeerreviewmaketitle




\section{Introduction}\label{sec:intro}

Cooperation of a large number of possibly small-sized robots, called \emph{robotic {\colg swarm}}, will play a significant role in complex missions that existing operational concepts using a few large robots could not deal with \cite{Shin2014a}.
Even if every single robot (or called \emph{agent}) in a swarm is incapable of accomplishing a task alone, 
their cooperation will lead to successful outcomes \cite{Khamis2015, Jevtic2012, Sahin2005, Dorigo2014}.
The possible applications include environmental monitoring \cite{Barton2013}, ad-hoc network relay \cite{Bekmezci2013}, 
disaster management \cite{Erdelj2017}, cooperative radar jamming \cite{Jang2017}, to name a few.

Due to the large cardinality of a swarm robot system, however, 
it is infeasible for human operators to supervise each {\colg agent} directly, but needed to entrust the swarm with certain levels of decision-makings (e.g., task allocation, path planning, and individual control). 
Thereby, {\col what only remains} is to provide a high-level mission description, which is manageable for a few or even a single human operator. 
Nevertheless, there still exist various challenges in the autonomous decision-making of robotic swarms. 
{\col Among} them, this paper addresses a task allocation problem 
where the number of agents is higher than that of tasks: how to partition a set of agents into subgroups and assign the subgroups to each task. 
In the problem, it is assumed that each agent can be assigned to at most one task, whereas each task may require multiple agents: this case falls into ST-MR (single-task robot and multi-robot task) category \cite{Gerkey2004, Korsah2013}.

According to \cite{Sahin2005, Dorigo2014, Bandyopadhyay2017, Brambilla2013, Johnson2011}, decision-making frameworks for {\col a robotic swarm should be} 
{\col \emph{decentralized} (i.e.,} the desired collective {\col behavior} can be achieved by individual agents {\col relying} on local information), 
\emph{scalable}, \emph{predictable} {\col (e.g., regarding convergence performance and outcome quality), and} 
\emph{adaptable} to dynamic environments (e.g., {\col unexpected} elimination or addition of agents or tasks).  
{\col Moreover, the frameworks are also desirable to be} \emph{robust} {\col against} asynchronous environments {\col because, d}ue to the large cardinality of the system and its {\col decentralization}, it is very challenging for {\col every} agent to {\col behave} synchronously. 
For {\col synchronization} in practice, ``artificial delays and extra communication must be built into the framework'' \cite{Johnson2011}, which may cause considerable inefficiency on the system. 
{\col In addition, it is also preferred to be capable of} accommodating different interests of agents (e.g., different swarms operated by different {\col organization}s \cite{Clark2009}).
In this paper, 
we propose a novel decision-making framework based on hedonic games \cite{Dreze1980, Banerjee2001, Bogomolnaia2002}. 
{\colg The task allocation problem considered is modeled as}
a coalition-formation game {\colg where} self-interest agents are willing to form coalitions to improve their own interests. 
The objective of this game is to find a \emph{Nash stable} partition, which is a social agreement where all the agents agree with the current task assignment.
Despite any possible conflicts between the agents, this paper shows that if they have \emph{social inhibition}, then a Nash stable partition can always be determined within polynomial times in the proposed framework and all the desirable characteristics mentioned above can be achieved. 
Furthermore, we {\col analyze} the lower bound of the outcome's suboptimality and show that 50\% is {\colg at least} guaranteed for a particular case. 
Various settings of numerical experiments validate that the proposed framework is scalable, {\colg adaptable}, and robust even in asynchronous environments. 

This paper is {\col organize}d as follows. 
Section \ref{sec:GRAPE_related_work} reviews existing literature on {\col decentralized} task allocation approaches 
and introduces a recent finding in hedonic game{\colg s} that inspires this study. 
Section \ref{GRAPE} proposes {\colg our} decision-making framework, named \emph{GRAPE}, and analytically proves the existence of and the polynomial-time convergence to a Nash stable partition. 
Section \ref{Analysis} discusses the framework's algorithmic complexity, suboptimality, adaptability, {\colg and robustness}. 
Section \ref{sec:min_rqmt} {\colg shows} that the framework can also address a task allocation problem in which each task may need a certain number of agents for completion. 
Numerical simulations in Section \ref{Results} confirm that the proposed framework holds all the desirable characteristics. 
Finally, concluding remarks are followed in Section \ref{Conclusion}.

\section{Related Work}\label{sec:GRAPE_related_work}

\subsection{{\col Decentralized} Coordination of Robotic Swarms}\label{sec:GRAPE_literature}

Existing approaches for task allocation problems can be {\col categorize}d into two branches, depending on how agents eventually reach a converged outcome: \emph{orchestrated} and \emph{(fully) self-{\col organize}d} approaches \cite{Brutschy2014}. 
In the former, additional mechanism such as negotiation and voting model is imposed so that some agents can be worse off if a specific condition is met (e.g., the global utility is better off).
Alternatively, in self-{\col organize}d approaches, each agent simply makes a decision without negotiating with {\colg the} other agents. 
The latter generally induce less resource consumption in communication and computation \cite{Kalra2006}, 
and hence they are preferable in terms of scalability. 
On the other hand, the former usually provide a better quality of solutions with respect to the global utility, and a certain level of suboptimality could be guaranteed \cite{Zhang2013, Choi2009, Segui-gasco2015}. 
A comparison result between them \cite{Kalra2006} presents that as the available information to agents becomes local, the latter becomes to outperform the former.
In the following, we particularly review existing literature on self-{\col organize}d approaches because, for large-scale multiple agent systems, scalability is {\colg at least} essential and it is realistic to regard that the agents only know their local information but instead the global information. 

Self-{\col organize}d approaches can be {\col categorize}d into \emph{top-down approaches} and \emph{bottom-up approaches} according to which level (i.e., between an ensemble and individuals) is mainly focused on. 
Top-down approaches {\col emphasize} developing a macroscopic model for the whole system. 
For {\col instance}, population fractions associated with given tasks are represented as states, and the dynamics of the population fractions {\colg is} {\col modeled} by Markov chains \cite{Acikmese2015, Chattopadhyay2009, Demir2015, Bandyopadhyay2017} or differential equations \cite{Berman2009, Halasz2007, Hsieh2008,  Mather2011, Prorok2016c}. 
Given a desired fraction distribution over the tasks, agents can converge to the desired status by following local decision policies (e.g., the associated rows or columns of the current Markov matrix). 
One advantage of using top-down approaches is 
{\col predictability of average emergent behavior with regard to}  
convergence speed and the quality of a stable outcome (i.e., how well the agents converge to the desired fraction distribution).    
{\colg However,} such prediction, to the best of our knowledge, can be made mainly numerically.
Besides, as top-down generated control policies regulate agents,
it may be difficult to accommodate each agent's individual preference. 
Also, each agent may have to physically move around according to its local policy during the entire decision-making process, which  may cause waste of time and energy costs in the transitioning. 

Bottom-up approaches focus on designing each agent's individual rules (i.e., microscopic models) that eventually lead to a desired emergent {\col behavior}. 
Possible actions of a single agent can be {\col modeled} by a finite state machine \cite{Labella2006}, and {\colg a} change of {\col behavior} occurs according to a probabilistic threshold model \cite{Castello2014}. 
A threshold value in the model {\col determines} the decision boundary {\col between two} motion{\col s.}  This value is adjustable based on an agent's past experiences such as {\colg the} time spent for working a task \cite{Brutschy2014, Kurdi2016}, 
the success/failure rates \cite{Labella2006, Liu2007a}, and direct communication from a central unit \cite{Castello2014}. 
This feature can improve system adaptability, and may have a potential to incorporate each agent's individual interest if required. 
However, it was shown in  \cite{Liu2010b, Martinoli2004, Lerman2005, Correll2006, Liu2007a, Prorok2011, Kanakia2016} that, 
to predict or evaluate an emergent performance of a swarm {\col utiliz}ing bottom-up approaches, a macroscopic model for the swarm is eventually required to be developed by abstracting the microscopic models.

\subsection{Hedonic Games}
\emph{Hedonic games} \cite{Dreze1980, Banerjee2001, Bogomolnaia2002} 
model a conflict situation where self-interest agents are willing to form coalitions to improve their own interests. 
\emph{Nash stability} \cite{Bogomolnaia2002} plays a key role since it yields a social agreement {\col among} the agents even without having any negotiation. 
Many researchers have investigated conditions under which a Nash stable partition is guaranteed to exist and to be determined \cite{Bogomolnaia2002, Dimitrov2006, Darmann2012, Darmann2015}. 
{\col Among} them, the works in \cite{Darmann2012, Darmann2015} mainly addressed an \emph{anonymous hedonic game}, in which each agent considers the size of a coalition to which it belongs instead of the identities of the members. 
Recently, Darmann \cite{Darmann2015} showed that selfish agents who have \emph{social inhibition} (i.e., preference toward a coalition with a fewer number of members) could converge to a Nash stable partition in an anonymous hedonic game. 
The author also proposed a {\col centralized} recursive algorithm that can find a Nash stable partition within $O( n_a^2 \cdot n_t)$ of iterations. Here, $n_a$ is the number of agents and $n_t$ is that of tasks. 

\subsection{Main Contributions}
Inspired by the recent breakthrough of \cite{Darmann2015}, 
we propose a novel {\col decentralized} framework that models the task allocation problem considered as an anonymous hedonic game. 
The proposed framework is a self-{\col organize}d approach in which agents make decisions according to its local policies (i.e., individual preferences). 
Unlike top-down or bottom-up approaches reviewed in the previous section, which primarily concentrate on designing agents' decision-making policies either macroscopically or microscopically, our work instead focuses on investigating and exploiting advantages from socially-inhibitive agents, while simply letting them greedily behave according to their individual preferences. 
Explicitly, the main contributions of this paper are as follows:
\begin{enumerate}
\item This paper shows that selfish agents with social inhibition, which we refer to as \emph{SPAO} preference (Definition \ref{SPAO}), can reach a Nash stable partition within less algorithmic complexity compared with \cite{Darmann2015}: $O(n_a^2)$ of iterations are required\footnote{Note that the definition of \emph{iteration} is described in Definition \ref{def:iteration}. This comparison assumes the fully-connected communication network because the algorithm in \cite{Darmann2015} is {\col centralized}.}. 
\item We provide a {\col decentralized} algorithm, which is executable under a strongly-connected communication network of agents and even in asynchronous environments. Depending on the network assumed, the algorithmic complexity may be additionally increased by $O(d_G)$, where $d_G < n_a$ is the graph diameter of the network. 
\item This paper {\col analyze}s the suboptimality of a Nash stable partition in term of the global utility. 
We firstly present a mathematical formulation to compute the suboptimality lower bound by using the information of a Nash stable partition and agents' individual utilities. 
Furthermore, we additionally show that 50\% of suboptimality can be at least guaranteed if the social utility for each coalition is defined as a non-decreasing function with respect to the number of members in the coalition. 
\item Our framework can accommodate different agents with different interests as long as their individual preferences hold SPAO.
\item Through various numerical experiments, it is confirmed that the proposed framework is scalable, fast adaptable to environmental changes, and robust even in a realistic situation where 
some agents are temporarily unable to proceed a decision-making procedure and communicate with {\colg the} other agents during a mission. 
\end{enumerate}



\begin{table}[h]
\renewcommand{\arraystretch}{1.3}
\caption{Nomenclature}
\label{nomenclature}
\centering
\begin{tabular}{p{0.4in} p{2.7in}}
\hline
Symbol & Description \\
\hline
\hline
$\mathcal{A}$	 		& a set of $n_a$ agents \\
$a_i$		& the $i$-th agent\\
$\mathcal{T}^*$ 		& a set of $n_t$ tasks \\
$t_j$			& the $j$-th task\\
$t_{\phi}$ 		& the void task (i.e., not to work any task) \\
$\mathcal{T}$	 		& a set of tasks, $\mathcal{T}=\mathcal{T}^* \cup \{t_{\phi}\}$\\
$(t_j,p)$		& a task-coalition pair (i.e. to do task $t_j$ with $p$ participants)\\
$\mathcal{X}$			& the set of task-coalition pairs, $\mathcal{X}=\mathcal{X}^* \cup \{t_{\phi}\}$, \\ & where $\mathcal{X}^* = \mathcal{T}^* \times \{1,2,...,n_a\}$\\
$\mathcal{P}_i$		& agent $a_i$'s preference relation over $\mathcal{X}$\\ 
$\succ_i$		& the strong preference of agent $a_i$\\
$\sim_i$		& the indifferent preference of agent $a_i$\\ 
$\succeq_i$ 	& the weak preference of agent $a_i$\\
$\Pi$			& a \emph{partition}: a disjoint set that partitions the agent set $\mathcal{A}$, \\ & $\Pi = \{S_1,S_2,...,S_{n_t},S_\phi\}$\\
$S_j$		& the (task-specific) coalition for $t_j$\\
$\Pi(i)$		& the index of the task to which agent $a_i$ is assigned given $\Pi$ \\
$d_G$		& the graph diameter of the agent communication network\\
${\colg \mathcal{N}_i}$	&  The {\colg neighbor agent set of agent $a_i$ given a network}\\
\hline
\end{tabular}
\end{table}

\section{GRoup Agent Partitioning and placing Event}\label{GRAPE}

\subsection{Problem Formulation}\label{sec:MRTA}
Let us first introduce the multi-robot task allocation problem considered in this paper and underlying assumptions. 
\begin{problem}\label{prob_basic}
Suppose that there exist a set of $n_a$ agents $\mathcal{A} =\{a_1, a_2, ... , a_{n_a}\}$ and
a set of tasks $\mathcal{T}=\mathcal{T}^* \cup \{t_\phi\}$, where ${ \mathcal{T}^*=\{t_1, t_2, ... , t_{n_t}\} }$ is a set of $n_t$ tasks and $t_\phi$ is \emph{the void task} (i.e., not to perform any task). 
Each agent $a_i$ has \emph{the individual utility} $u_i: \mathcal{T} \times {|\mathcal{A}|} \rightarrow \mathbb{R}$, which is a function of the task to which the agent is assigned and the number of {\colg its} co-working agents (including itself) {\colg $p \in \{1,2,...,n_a\}$} (called \emph{participants}).
{\colg The individual utility for $t_{\phi}$ is zero regardless of the participants.}
Since every agent is considered to have limited capabilities to finish a task alone, the agent can be assigned to at most one task.  
The objective of this task allocation problem is to find an assignment that {\col maximize}s \emph{the global utility}, which is the sum of individual utilities of the entire agents. 
The problem described above is defined as follows:
\begin{equation}\label{eqn:obj_ftn}
\max_{\{x_{ij}\}} \sum_{\forall a_i \in \mathcal{A}} \sum_{\forall t_j \in \mathcal{T}} u_{i}(t_j, p) x_{ij} {\colg ,}
\end{equation}
subject to
\begin{equation}
\sum_{\forall t_j \in \mathcal{T}} x_{ij} \le 1{\colg ,} \quad \forall a_i \in \mathcal{A}{\colg ,}
\end{equation}
\begin{equation}
x_{ij} \in \{0,1\}{\colg ,} \quad \forall a_i \in \mathcal{A}, \forall t_j \in \mathcal{T}{\colg ,}
\end{equation}
where $x_{ij}$ is a binary decision variable that indicates whether or not task $t_j$ is assigned to agent $a_i$.
\end{problem}

The term \emph{social utility} is defined as the sum of individual utilities within any agent group.

\begin{assumption}[\emph{Homogeneous agents with limited capabilities}]\label{assum:agents}
This paper considers a large-scale multi-robot system of homogeneous agents since the realisation of a swarm can be in general achieved through mass production \cite{Sahin2005}. 
Therefore, each individual utility $u_i$ is concerned with the cardinality of the agents working for the task. 
Note that agents in this paper may have different preferences with respect to the given tasks, e.g., for an agent, a spatially closer task is more preferred, whereas this may not be the case for another agent.
Besides, noting that ``mass production {\col favors} robots with fewer and cheaper components, resulting in lower cost but also reduced capabilities\cite{Rubenstein2014}", 
it is also assumed that each agent can be only assigned to perform at most a single task. 
According to \cite{Gerkey2004}, such a robot is called a \emph{single-task} (ST) robot. 
\end{assumption}

\begin{assumption}[\emph{Agents' communication}]\label{assum:agents_comm}
The communication network of the entire agents is at least \emph{strongly-connected}, 
{\colb i.e., there exists a directed communication path between any two arbitrary agents.} 
Given a network, $\mathcal{N}_i$ denotes a set of {\col neighbor} agents for agent $a_i$. 
\end{assumption}

\begin{assumption}[\emph{Multi-robot-required tasks}]\label{assum:tasks}
Every task is a \emph{multi-robot} (MR) task, meaning that the task may require multiple robots \cite{Gerkey2004}. 
For now, we assume that each task can be performed even by a single agent although it may take a long time.
However, in Section \ref{sec:min_rqmt}, we will also address a particular case in which some tasks need at least a certain number of agents for completion.
\end{assumption}

\begin{assumption}[\emph{Agents' pre-known information}]\label{assum:agents_util}
Every agent $a_i$ only knows its own individual utility $u_i(t_j,p)$ with regard to every task $t_j$, while not being aware of those of {\colg the} other agents. 
Through communication, however, they can notice which agent currently choses which task, i.e., \emph{partition} (Definition \ref{def:partition}). 
Note that the agents do not necessarily have to know the true partition information at all the time. Each agent owns its locally-known partition information. 
\end{assumption}

\subsection{Proposed Game-theoretical Approach: GRAPE}
Let us transform Problem \ref{prob_basic} into an anonymous hedonic game event where 
every agent selfishly tends to join a coalition according to its preference. 


\begin{definition}[\emph{GRAPE}]
\label{game}
An instance of \emph{GRoup Agent Partitioning and placing Event} (GRAPE) is a tuple $(\mathcal{A}, \mathcal{T}, \mathcal{P})$ 
that consists of 
(1) ${ \mathcal{A} =\{a_1, a_2, ... , a_{n_a}\} }$, a set of  $n_a$ agents; 
(2) ${ \mathcal{T}=\mathcal{T}^* \cup \{t_\phi\} }$, a set of tasks; 
and (3) ${ \mathcal{P}=(\mathcal{P}_1, \mathcal{P}_2, ... , \mathcal{P}_{n_a}) }$, an {$n_a$}-tuple of preference relations of the agents. 
For agent $a_i$, $\mathcal{P}_i$ describes its \emph{preference relation} over the set of task-coalition pairs ${ \mathcal{X}=\mathcal{X}^* \cup \{t_\phi\} }$, 
where ${ \mathcal{X}^*=\mathcal{T}^*\times \{1,2,...,n_a\} }$;
a task-coalition pair $(t_j,p)$ is interpreted as ``to do task $t_j$ with $p$ participants''. 
For any task-coalition pairs ${ x_1, x_2 \in \mathcal{X} }$, 
${ x_1 \succ_i x_2 }$ implies that agent $a_i$ strongly prefers $x_1$ to $x_2$, and 
${ x_1 \sim_i x_2 }$ means that the preference regarding $x_1$ and $x_2$ is indifferent. 
Likewise, $\succeq _i$ indicates the weak preference of agent $a_i$.
\end{definition}

Note that agent $a_i$'s preference relation can be derived from its individual utility $u_i(t_j, p)$ in Problem \ref{prob_basic}. 
For {\col instance}, given that $u_i(t_1,p_1) > u_i(t_2,p_2)$, it can be said that $(t_1,p_1) \succ_i (t_2,p_2)$.

\begin{definition}[\emph{Partition}]\label{def:partition}
Given an instance $(\mathcal{A},\mathcal{T},\mathcal{P})$ of GRAPE, 
a \emph{partition} is defined as a set $\Pi = \{S_1,S_2,...,S_{n_t},S_\phi \}$
that disjointly partitions the agent set $\mathcal{A}$. 
Here, $S_j \subseteq \mathcal{A}$ is the \emph{(task-specific) coalition} for executing task $t_j$ 
such that $\cup^{n_t}_{j=0}S_j=\mathcal{A}$ and $S_j \cap S_k = \emptyset$ for $j \neq k$. 
$S_\phi$ is the set of agents who choose the void task $t_\phi$. 
Note that this paper interchangeably uses $S_0$ to indicate $S_\phi$. 
%
Given a partition $\Pi$, $\Pi(i)$ indicates the index of the task to which agent $a_i$ is assigned.
For {\col example}, $S_{\Pi(i)}$ is the coalition that the agent belongs to, i.e., ${S_{\Pi(i)}= \{S_j \in \Pi \mid a_i \in S_j\}}$. 
\end{definition}

The objective of GRAPE is to determine a stable partition that all the agents agree {\colg with}. 
In this paper, we seek for a \emph{Nash stable} partition, which is defined as follows:

\begin{definition}[\emph{Nash stable}]
\label{Nash_stable}
A partition $\Pi$ is said to be \emph{Nash stable} 
if, for every agent ${ a_i \in \mathcal{A} }$, it holds that ${ (t_{\Pi(i)}, |S_{\Pi(i)}|) \succeq_i  (t_j, |S_j \cup \{a_i\}|)  }$, $\forall {S_j \in \Pi}$. 
\end{definition} 

In other words, in a Nash stable partition, every agent prefers its current coalition to joining any of the other coalitions. 
Thus, every agent does not have any conflict within this partition, and no agent will not unilaterally deviate from its current decision.

\begin{remark}[\emph{An advantage of Nash stability: low communication burden on agents}]
The rationale behind the use of Nash stability {\col among} various stable solution concepts in hedonic games \cite{Sung2007, Dreze1980, Karakaya2011, Aziz2012a} is that it can reduce communication burden between agents required to reach a social agreement. 
In the process of converging to a Nash stable partition, an agent does not need to get any permission from {\colg the} other agents 
when it is willing to deviate. This property may not be the case for the other solution concepts. 
Therefore, each agent is only required to notify its altered decision without any negotiation. 
This fact can reduce {\colg inter-agent} communication 
in the proposed approach. 
\end{remark}

\subsection{SPAO Preference: Social Inhibition}
This section introduces the key condition, called \emph{SPAO}, that enables our proposed approach to provide all the desirable properties described in Section \ref{sec:intro}, 
and then explains its implications. 


\begin{definition}[\emph{SPAO}]
\label{SPAO}
Given an instance ${ (\mathcal{A},\mathcal{T},\mathcal{P}) }$ of GRAPE, it is said that 
the preference relation of agent $a_i$ with respect to task $t_j$ is \emph{SPAO (Single-Peaked-At-One)}
if it holds that,
for every ${ (t_j,p) \in \mathcal{X}^* }$, 
${ (t_j,p_1) \succeq _i (t_j,p_2) }$ for any ${p_1,p_2 \in \{1,...,n_a\} }$ such that ${ p_1 < p_2 }$. 
Besides, we say that an instance ${ (\mathcal{A},\mathcal{T},\mathcal{P})}$ of GRAPE is SPAO 
if the preference relation of every agent in $\mathcal{A}$ with respect to every task in ${\mathcal{T}^*}$ is SPAO. 
\end{definition}

For an example, suppose that $\mathcal{P}_i$ is such that 
\begin{equation*}
(t_1,1) \succ_i (t_1,2) \succeq_i (t_1,3) \succ_i (t_2,1)  \sim_i (t_1,4) \succ_i (t_2,2).
\end{equation*}
This preference relation indicates that agent $a_i$ has $(t_1,1) \succ_i (t_1,2) \succeq_i (t_1,3) \succ_i (t_1,4)$ for task $t_1$, 
and $(t_2,1) \succ_i (t_2,2)$ for task $t_2$. 
According to Definition \ref{SPAO}, the preference relation for each of the tasks holds SPAO. 
For another example, given that
\begin{equation*}
(t_1,1) \succ_i (t_1,2) \succeq_i (t_1,3) \succ_i (t_2,2)  \sim_i (t_1,4) \succ_i (t_2,1), 
\end{equation*}
the preference relation regarding task $t_1$ holds SPAO, whereas this is not the case for task $t_2$ because of $(t_2,2)  \succ_i (t_2,1)$.

This paper only considers the case in which every agent has SPAO preference relations regarding all the given tasks. 
{\colg Such a}gents prefer to execute a task with smaller number of collaborators, namely, they have \emph{social inhibition}.

\begin{remark}[\emph{Implications of SPAO}]
SPAO implies that an agent's individual utility should be a monotonically decreasing function with respect to the size of a coalition. 
In practice, SPAO can often emerge.  
For instance, experimental and simulation results in \cite[Figures 3 and 4]{Guerrero2012} show that the total work capacity resulted from cooperation of multiple robots does not proportionally increase due to interferences of the robots. 
In such a \emph{non-superadditive} environment \cite{Shehory1999}, 
assuming that an agent's individual work efficiency is considered as its individual utility, 
the individual utility monotonically drops as the number of collaborators enlarges even though the social utility is increased.  
For another example, SPAO also arises when individual utilities are related with shared-resources.  
As more agents use the same resource simultaneously, their individual productivities become diminished (e.g., traffic affects travel times \cite{Nam2015} \cite[Example 3]{Johnson2016}). 
As the authors in \cite{Shehory1999} pointed out, a non-superadditive case is more realistic than a superadditive case: agents in a superadditive environment always attempt to form the grand coalition whereas those in a non-superadditive case are willing to reduce unnecessary costs. 
Note that social utility functions are not restricted so that they can be either monotonic or non-monotonic. 
\end{remark}

\begin{remark}[\emph{Cooperation of selfish agents {\colg with different interests}}]
The proposed framework can accommodate selfish agents who greedily follow their individual preferences as long as the preferences hold SPAO. 
This implies that the framework may be {\col utilized} for a combination of swarm systems from different organisations under the condition that the multiple systems satisfy SPAO. 
\end{remark}

\subsection{Existence of and Convergence to a Nash Stable Partition}\label{sec:existence_NS}

Let us prove that 
if an instance of GRAPE holds SPAO, there always exists a Nash stable partition and it can be found within polynomial time.   


\begin{definition}[\emph{Iteration}]\label{def:iteration}
This paper uses the term \emph{iteration} to represent an iterative stage in which an arbitrary agent compares the set of selectable task-coalition pairs given an existing partition, and then determines whether or not to join another coalition including the void task one. 
\end{definition}

\begin{assumption}[\emph{Mutual exclusion algorithm}]\label{assum:mutex}
We assume that, at each iteration, a single agent exclusively makes a decision and updates the {\colg current} partition if necessary. 
This paper refers to this agent as the \emph{deciding agent} at the iteration. 
Based on the resultant partition, another deciding agent also performs the same process at the next iteration, 
and this process continues until every agent does not deviate from a specific partition, which is, in fact, a Nash stable partition. 
To implement this algorithmic process in practice, the agents need a \emph{mutual exclusion} (or called \emph{mutex}) algorithm to choose the deciding agent at each iteration. 
In this section, for simplicity of description, we assume that all the agents are fully-connected, by which they somehow select and know the deciding agent. 
However, in Section \ref{sec:algorithm}, we will present a distributed mutex algorithm that enables the proposed approach to be executed under a strongly-connected communication network even in an asynchronous manner. 
\end{assumption}


\begin{lemma}
\label{lemma_1}
Given an instance ${(\mathcal{A},\mathcal{T},\mathcal{P})}$ of GRAPE that is SPAO, 
suppose that a new agent $a_r \notin \mathcal{A}$ holding a SPAO preference relation with regard to every task in $\mathcal{T}$ joins ${(\mathcal{A},\mathcal{T},\mathcal{P})}$ 
in which a Nash stable partition 
is already established.
Then, the new instance ${(\tilde{\mathcal{A}},\mathcal{T},\mathcal{P})}$, where ${\tilde{\mathcal{A}} = \mathcal{A} \cup \{a_r\} }$, also 
(1) satisfies SPAO; 
(2) contains a Nash stable partition; 
and (3) the maximum number of iterations required to re-converge to a Nash stable partition is ${|\tilde{\mathcal{A}}|}$.      
\end{lemma}

\begin{proof}

Given a partition $\Pi$, {\colg}{for agent $a_i$,}
\emph{the number of {\colg additional co-workers tolerable in} its coalition} is defined as:
{\colg 
\begin{equation}
\resizebox{.48\textwidth}{!}{$
\Delta_{\Pi(i)} := \min\limits_{{S}_j \in \Pi \setminus \{{S}_{\Pi(i)}\}} \max\limits_{\Delta \in \mathbb{Z}} \big\{ \Delta \mid (t_{\Pi(i)},|{S}_{\Pi(i)}| + \Delta) \succeq_i (t_j,|{S}_j \cup \{a_i\}|) \big\}.
$}
\end{equation}
}
%
Due to the SPAO preference relation, this value satisfies the following characteristics:
(a) if $\Pi$ is Nash stable, for every agent $a_i$, it holds that $\Delta_{\Pi(i)} \ge 0$;
(b) if $\Delta_{\Pi(i)} < 0$, then agent $a_i$ is willing to deviate to {\colg another} coalition at a next iteration;
and (c) for the agent $a_i$ who deviated at the last iteration and updated the partition as $\Pi'$, it holds that $\Delta_{\Pi'(i)} \ge 0$. 

From Definition \ref{SPAO}, it is clear that the new instance ${ (\tilde{\mathcal{A}},\mathcal{T},\mathcal{P}) }$ still holds SPAO. 
Let ${\Pi_0}$ denote a Nash stable partition in the original instance ${ (\mathcal{A},\mathcal{T},\mathcal{P}) }$. 
When a new agent $a_r \notin \mathcal{A}$ decides to execute one of {\colg the} tasks in $\mathcal{T}$ and creates a new partition ${ {\Pi}_1}$, 
it holds that $\Delta_{{\Pi}_1(r)} \ge 0$, as shown in (c). 
If there is no existing agent $a_q \in \mathcal{A}$ whose $\Delta_{{\Pi}_1(q)} < 0$, then the new partition ${ {\Pi}_1}$ is Nash stable.

   
Suppose that there exists at least an agent $a_q$ whose $\Delta_{{\Pi}_1(q)} < 0$. 
Then, the agent must be one of {\colg the existing} members in the coalition that agent $a_r$ selected in the last iteration. 
As agent $a_q$ moves to another coalition and creates a new partition ${ {\Pi}_2}$, the previously-deviated agent $a_r$ {\colg holds} $\Delta_{{\Pi}_2(r)} \ge 1$.
In other words, an agent who deviates to a coalition and expels one of the existing agents in that coalition will not deviate again even if another agent joins the coalition in a next iteration.     
This implies that at most ${|\tilde{\mathcal{A}}|}$ of iterations are required to hold $\Delta_{\tilde{{\Pi}}(i)} \ge 0$ for every agent $a_i \in \tilde{\mathcal{A}}$, where the partition $\tilde{\Pi}$ is Nash stable. 
\end{proof}


Lemma \ref{lemma_1} is essential 
not only for the existence of and convergence to a Nash stable partition 
but also for fast adaptability to dynamic environments.

\begin{theorem}[\emph{Existence}]
\label{NASH}
If ${ (\mathcal{A},\mathcal{T},\mathcal{P}) }$ is an instance of GRAPE holding SPAO, then a Nash stable partition always exists. 
\end{theorem}

\begin{proof}
This theorem will be proved by induction. 
Let $M(n)$ be the following mathematical statement: {\colg for} $|\mathcal{A}| = n$, if an instance ${ (\mathcal{A},\mathcal{T},\mathcal{P}) }$ of GRAPE is SPAO, then there exists a Nash stable partition.  

\emph{Base case}: When ${n=1}$, there is only one agent in an instance. This agent is allowed to participate in its most preferred coalition, and the resultant partition is Nash stable.
Therefore, $M(1)$ is true. 

\emph{Induction hypothesis}: Assume that ${M(k)}$ is true for a positive integer $k$ such that ${|\mathcal{A}|=k}$. 

\emph{Induction step}: Suppose that a new agent ${a_i \notin \mathcal{A}}$ whose preference relation regarding every task in $\mathcal{T}$ is SPAO joins the instance ${  (\mathcal{A},\mathcal{T},\mathcal{P})  }$. 
This induces a new instance ${ (\tilde{\mathcal{A}},\mathcal{T},\mathcal{P}) }$ where ${\tilde{\mathcal{A}} = \mathcal{A} \cup \{a_i\} }$ and ${ |\tilde{\mathcal{A}}|=k+1}$. 
From Lemma \ref{lemma_1}, it is clear that the new instance also satisfies SPAO and has a Nash stable partition ${ \tilde{\Pi} }$. 
Consequently, ${M(k+1)}$ is true. 
By mathematical induction, ${M(n)}$ is true for all positive integers $n \ge 1$. 
\end{proof}


\begin{theorem}[\emph{Convergence}]
\label{Nash_poly}
If ${ (\mathcal{A},\mathcal{T},\mathcal{P}) }$ is an instance of GRAPE holding SPAO, 
then the number of iterations required to determine a Nash stable partition is at most ${ |\mathcal{A}|\cdot(|\mathcal{A}|+1)/2}$. 
\end{theorem}

\begin{proof}
Suppose that, given a Nash stable partition in an instance where there exists only one agent, 
we add another arbitrary agent and find a Nash stable partition for this new instance, 
and repeat the procedure until all the agents in $\mathcal{A}$ are included. 
From Lemma \ref{lemma_1}, if a new agent joins an instance 
in which the current partition is Nash stable, then the maximum number of iterations required to find a new Nash stable partition is the number of the existing agents plus {\colg one}. 
Therefore, it is trivial that the maximum number of iterations to find a Nash stable partition of an instance ${ (\mathcal{A},\mathcal{T},\mathcal{P}) }$ is given as \begin{equation}
{\sum_{k=1}^{|\mathcal{A}|} k = |\mathcal{A}| \cdot (|\mathcal{A}|+1)/2}.
\end{equation}

Note that this {\colg polynomial-time convergence still holds} even if the agents are {\col initialize}d {\colg to} a random partition.
Suppose that we have the following setting: the entire agents $\mathcal{A}$ are firstly not movable from the existing partition, except a set of free agents $\mathcal{A'} \subseteq \mathcal{A}$; whenever the agents $\mathcal{A}'$ find a Nash stable partition $\Pi'$, one arbitrary agent in $a_r \in \mathcal{A} \setminus \mathcal{A}'$ additionally becomes liberated and deviates from the current coalition $S_{\Pi'(r)}$ to another coalition in $\Pi'$. 
In this setting, {\colg from} the viewpoint of the agents in $\mathcal{A}' \setminus S_{\Pi'(r)}$, the newly liberated agent is considered as a new agent as that in Lemma \ref{lemma_1}. 
Accordingly, we can still {\col utilize} the lemma for the agents in $\mathcal{A}' \setminus S_{\Pi'(r)} \cup \{a_r\}$. 
The agents also can find a Nash stable partition if one of them moves to $S_{\Pi'(r)}$ during the process, because, due to $a_r$, it became $\Delta_{\Pi'(i)} \ge 1$ for every agent $a_i \in S_{\Pi'(r)} \setminus \{a_r\}$. 
In a nutshell, the agents $\mathcal{A}' \cup \{a_r\}$ can converge to a Nash stable partition within $|\mathcal{A}' \cup \{a_r\}|$, {\colg which is equivalent to} Lemma \ref{lemma_1}. {\colg H}ence{\colg,} Theorem \ref{NASH} and this theorem are also valid for the case when the initial partition of the agents are randomly given. 
\end{proof}


\subsection{{\col Decentralized} Algorithm}\label{sec:algorithm}

In the previous section, it {\colg was} assumed that only one agent is somehow chosen to make a decision at each iteration {\colg under the fully-connected network}. 
On the contrary, in this section, we propose a {\col decentralized} algorithm, as shown in Algorithm \ref{algorithm}, in which every agent does decision making {\colg based its local information} and affects its {\col neighbors} simultaneously {\colg under a strongly-connected network.} 
Despite that, we show that {\colg Theorems \ref{NASH} and \ref{Nash_poly} still hold} thanks to our proposed distributed mutex subroutine shown in Algorithm \ref{algorithm:async}. 
The details of the {\col decentralized} main algorithm are as follows. 

\begin{algorithm}
\caption{Decision-making algorithm for each agent $a_i$}\label{algorithm}
\begin{algorithmic}[1]
\Statex \emph{// Initialisation}
\State $\mathsf{satisfied} \leftarrow false$; $r^i \leftarrow 0$; $s^i \leftarrow 0$ \label{line:local_variable_1}
\State $\Pi^i \leftarrow \{S_{\phi} = \mathcal{A}, S_j = \phi \ \forall t_j \in \mathcal{T}\}$\label{line:local_variable_2}
\Statex \emph{// Decision-making process begins}
	\While{$true$} 

	\Statex \emph{// Make a new decision if necessary}
	\If{$\mathsf{satisfied} = false$} \label{line:decision_making}

	\State $(t_{j*},|S_{j*}|) \leftarrow \arg\max_{\forall S_j \in \Pi^i} (t_j, |S_j \cup \{a_i\}|)$ \label{line:choose_best}
	\If{$(t_{j*},|S_{j*}|) \succ_i (t_{\Pi^i (i)},|S_{\Pi^i (i)}|)$} \label{line:decision_1}
	\State Join $S_{j*}$ and update $\Pi^i$
	\State $r^i \leftarrow r^i + 1$ \label{line:iteration_increase}
	\State $s^i \in \mathrm{unif}[0,1]$
	\EndIf \label{line:decision_2}
	\State $\mathsf{satisfied} = true$ \label{line:satisfied}
	\EndIf \label{line:decision_making2}
	
	\Statex \emph{// Broadcast the local information to {\col neighbor} agents}	\
	\State Broadcast $M^i = \{r^i, s^i, \Pi^i\}$ and receive	 $M^k$ from its\label{line:communication}
	\Statex \quad \ \ {\col neighbor}s $\forall a_k \in \mathcal{N}_i$ 

	\Statex \emph{// Select the valid partition from all the received messages}
	\State Construct $\mathcal{M}^i_{rcv}=\{M^i, \forall M^k\}$ 
	\State $\{r^i, s^i, \Pi^i\}, \mathsf{satisfied} \leftarrow$ \Call{D-Mutex}{$\mathcal{M}^i_{rcv}$}\label{line:decision_making_end}

	\EndWhile

\end{algorithmic}
\end{algorithm}

Each agent $a_i$ has local variables such as $\Pi^i$, $\mathsf{satisfied}$, $r^i$, and $s^i$ (Line \ref{line:local_variable_1}--\ref{line:local_variable_2}). 
Here, $\Pi^i$ is the agent's locally-known partition; 
$\mathsf{satisfied}$ is a binary variable that indicates whether or not the agent {\colb is satisfied} with $\Pi^i$ {\colb such that it does not want to deviate from its current coalition}; 
$r^i \in \mathbb{Z}^+$ is an integer variable to represent how many times $\Pi^i$ has evolved (i.e., the number of iterations happened for updating $\Pi^i$ until that moment); 
and $s^i \in [0,1]$ is a uniform-random variable that is generated when{\colg ever} $\Pi^i$ is newly updated (i.e., a random time stamp). 
Given $\Pi^i$, 
agent $a_i$ examines which coalition is the most preferred {\col among} others, assuming that {\colg the} other agents remain at the existing {\colg coalitions} (Line \ref{line:choose_best}). 
Then, the agent joins the newly found coalition if it is strongly preferred than {\colg its current} coalition. 
In this case, the agent updates $\Pi^i$ to reflect its new decision, increases $r^i$, and generates a new random time stamp $s^i$ (Line \ref{line:decision_1}--\ref{line:decision_2}).   
In any case, since the agent ascertained that the currently-selected coalition is the most preferred, 
the agent becomes satisfied with $\Pi^i$ (Line \ref{line:satisfied}). 
Then, agent $a_i$ generates {\colg and sends} a message $M^i := \{r^i, s^i, \Pi^i\}$ to {\colg its} {\col neighbor} agents, and vice versa (Line \ref{line:communication}).

Since every agent locally updates its locally-known partition simultaneously, 
one of the partitions should be regarded as if it were the partition updated by a deciding agent at the previous iteration. 
We refer to this partition as \emph{the valid partition} at the iteration.  
The distributed mutex subroutine in Algorithm \ref{algorithm:async} enables the agents to {\col recognize} the valid partition {\col among} all the locally-known current partitions even under a strongly-connected network and in asynchronous environments. 
Before executing this subroutine, each agent $a_i$ collects all the messages received from its {\col neighbor} agents ${\colg \forall a_k} \in \mathcal{N}_i$ (including $M^i$) as $\mathcal{M}^i_{rcv}= \{M^i, \forall M^k\}$.
Using this message set, the agent examines whether or not its own partition $\Pi^i$ is valid. 
If there exists any other partition $\Pi^k$ such that $r^k > r^i$, then the agent considers $\Pi^k$ more valid than $\Pi^i$. 
This also happens if {\colb $r^k = r^i$ and $s^k > s^i$, which indicates the case where $\Pi^k$ and $\Pi^i$ have evolved over the same amount of times, but the former has a higher time stamp}. 
Since $\Pi^k$ is considered as more valid, agent $a_i$ {\colg will need} to re-examine if there is a more preferred coalition given $\Pi^k$ in the next iteration. Thus, the agent sets $\mathsf{satisfied}$ as $false$ (Line \ref{alg_mutex:line1}--\ref{alg_mutex:line2} in Algorithm \ref{algorithm:async}). 
After {\colg completing this subroutine}, depending on $\mathsf{satisfied}$, each agent proceeds the decision-making process again (i.e., Line \ref{line:decision_making}--\ref{line:decision_making2} in Algorithm \ref{algorithm}) and/or just broadcasts the existing locally-known partition to its {\col neighbor} agents (Line \ref{line:communication} in Algorithm \ref{algorithm}).

\begin{algorithm}
\caption{Distributed Mutex Subroutine}\label{algorithm:async}
\begin{algorithmic}[1]

\Function{D-Mutex}{$\mathcal{M}^i_{rcv}$}
	\State $\mathsf{satisfied} \leftarrow true$
	\For{each message $M_k \in \mathcal{M}^i_{rcv}$} \label{alg_mutex:line1}
	\If{{\colg ($r^k > r^i$) or ($r^k = r^i$ \& $s^k > s^i$)}} 
		\State $r^i \leftarrow r^k$
		\State $s^i \leftarrow s^k$
		\State $\Pi^i \leftarrow \Pi^k$
		\State $\mathsf{satisfied} \leftarrow false$
	\EndIf 
	\EndFor \label{alg_mutex:line2}	
	\State \Return $\{r^i, s^i, \Pi^i\}$, $\mathsf{satisfied}$
\EndFunction
\end{algorithmic}
\end{algorithm}


In a nutshell, the distributed mutex algorithm makes sure that there is only one valid partition that dominates (or will finally dominate depending on the communication network) any other partitions. 
In other words, 
multiple partitions locally evolve, {\colg but one} of them only eventually survive {\colg as long as a strongly-connected network is given.
From each partition's viewpoint, it can be regarded as being evolved by a random sequence of the agents under the fully-connected network. 
Thus, the partition becomes Nash stable within the polynomial time as shown in Theorem \ref{Nash_poly}.} 
In an extreme case, we may encounter multiple Nash stable partitions {\colg at} the very last. 
Nevertheless, thanks to the mutex algorithm, one of them can be distributedly selected by the agents. 
{\colg All the features} imply that agents using Algorithm \ref{algorithm} can find a Nash stable partition in a {\col decentralized} manner {\colg and Theorems \ref{NASH} and \ref{Nash_poly} still hold.}

\section{Analysis}\label{Analysis}

%
%


\subsection{Algorithmic Complexity (Scalability)}\label{sec:algorithm_complexity}
Firstly, let us discuss about the running time for the proposed framework to find a Nash stable partition. 
This paper refers to a unit time required for each agent to proceed the main loop of Algorithm \ref{algorithm} (Line \ref{line:decision_making}-\ref{line:decision_making_end}) as a \emph{time step}. 
Depending on the communication network considered, especially if it is not fully-connected,
it may be possible that {\colg some of} the given agents have to execute this loop to just propagate {\colg their locally-known} partition information without affecting $r_i$ as Line \ref{line:iteration_increase}. 
Because this process also spends a unit time step, we call it as \emph{dummy iteration} to distinguish from a \emph{(normal) iteration}, which increases $r_i$.  

Notice that such dummy iterations happen {\colg sequentially} at most $d_G$ times before a normal iteration occurs, where $d_G$ is the graph diameter of the communication network. 
Hence, thanks to Theorem \ref{Nash_poly}, the total required time steps until finding a Nash stable partition is $O(d_G n_a^2)$. 
For the fully-connected network case, it becomes $O(n_a^2)$ because of $d_G = 1$. 
Note that this algorithmic complexity is less than that of the {\col centralized} algorithm, i.e., $O(n_a^2 \cdot n_t)$, in \cite{Darmann2015}.

Every agent at each iteration investigates ${n_t+1}$ of selectable task-coalition pairs including {\colg $t_{\phi}$} given a locally-known valid partition (as shown in Line \ref{line:choose_best} in Algorithm \ref{algorithm}). 
Therefore, the computational overhead for an agent is ${O(n_t)}$ per any iteration.  
With consideration of the total required time steps, 
the running time of the proposed approach for an agent can be bounded by ${O(d_G n_t n_a^2)}$. 
Note that the running time in practice can be much less than the bound since Theorem \ref{Nash_poly} was conservatively {\col analyze}d, as described in the following remark.

\begin{remark}[\emph{The number of required iterations in practice}]\label{remark:num_iter_practice}
Algorithm \ref{algorithm} allows the entire agents in $\mathcal{A}$ to {\colg be involved in} the decision-making process, 
whereas, in the proof for Theorem \ref{Nash_poly}, a new agent can be involved after a Nash stable partition of existing agents is found. 
Since agents using Algorithm \ref{algorithm} do not need to find every Nash stable partition for each subset of the agents, 
unnecessary iterations can be reduced. 
Hence, the number of required iterations in practice may become less than that shown in Theorem \ref{Nash_poly}, 
which is also supported by the experimental results in Section {\colg \ref{sec:result_scalability}}.
\end{remark}

{\col
Let us now discuss about the communication overhead for each agent per iteration. 
Given a network, agent $a_i$ should communicate with $|\mathcal{N}_i|$ of its neighbors, and the size of each message grows with regard to $n_a$.
Hence, the communication overhead of the agent is $O(|\mathcal{N}_i| \cdot  n_a )$.   
It could be quadratic if $|\mathcal{N}_i|$ increases in proportional to $n_a$. 
However, this would rarely happen in practice due to spatial distribution of agents and physical limits on communication such as range limitation.
Instead, $|\mathcal{N}_i|$ would be most likely saturated in practice.

\begin{remark}[\emph{Communication overhead vs. Running time}]\label{remark:comm_overhead}
To reduce the communication overhead, we may impose \emph{the maximum number of transactions per iteration}, denoted by $n_c$, on each agent. 
Even so, Theorems \ref{NASH} and \ref{Nash_poly} are still valid as long as the union of underlying graphs of the communication networks over time intervals becomes connected. 
However, in return, the number of dummy iterations may increase, so does the framework's running time. 
In an extreme case where $n_c = 1$ (i.e., unicast mode), dummy iterations may happen in a row at most $n_a$ times. 
Thus, the total required time steps until finding a Nash stable partition could be $O(n_a^3)$, whereas the communication overhead is $O(n_a)$. 
In short, the running time of the framework can be traded off against the communication overhead for each agent per iteration. 
\end{remark}
}

\subsection{Suboptimality}\label{sec:suboptimality}

This section investigates the \emph{suboptimality lower bound} (or can be called \emph{approximation ratio}) of the proposed framework in terms of the global utility, i.e., the objective function in Equation (\ref{eqn:obj_ftn}). 
Given a partition $\Pi$, the global utility value can be equivalently rewritten as
\begin{equation}
\label{Objective_function}
J = \sum_{\forall a_i \in \mathcal{A}} u_i(t_{\Pi(i)},|S_{\Pi(i)}|).
\end{equation}
Note that we can simply derive $\{x_{ij}\}$ for Equation (\ref{eqn:obj_ftn}) from $\Pi$ for Equation (\ref{Objective_function}), and vice versa. 
Let ${J_{GRAPE}}$ and $J_{OPT}$ represent 
the global utility of a Nash stable partition obtained by the proposed framework and the optimal value, respectively.  
This paper refers to the fraction of ${J_{GRAPE}}$ with respect to $J_{OPT}$ as the \emph{suboptimality} of GRAPE, denoted by $\alpha$, i.e.,  
\begin{equation}\label{eqn:approx_ratio}
\alpha :=  J_{GRAPE}/J_{OPT}.
\end{equation}

The lower bound of the suboptimality can be determined by the following theorem. 
\begin{theorem}[\emph{Suboptimality lower bound: general case}] 
\label{OPT_Bound}
Given a Nash stable partition ${\Pi}$ obtained by GRAPE, its suboptimality in terms of the global utility is lower bounded as follows:
\begin{equation}
\label{Eq_OPT_1}
\alpha \ge J_{GRAPE}/(J_{GRAPE}+\lambda),
\end{equation}
where 
\begin{equation}\label{eqn:lambda}
\lambda  \equiv  \sum_{\forall S_j \in \Pi} \max_{a_i \in \mathcal{A}, p \le |\mathcal{A}|} \big\{ p \cdot \big[ u_i(t_j,p) - u_i(t_j,|S_j \cup \{a_i\}|) \big] \big\} 
\end{equation}
\end{theorem}

\begin{proof}
Let ${\Pi^{*}}$ denote the optimal partition for the objective function in Equation (\ref{Objective_function}). 
Given a Nash stable partition $\Pi$, from Definition \ref{Nash_stable}, it holds that, $\forall a_i \in \mathcal{A}$,
\begin{equation}
\label{Eq_OPT_2}
	u_i (t_{\Pi(i)},|S_{\Pi(i)}|)  \ge  u_i ( t^{*}_{j \gets i} , | S_{j} \cup \{a_i\} | ) , 
\end{equation}
where ${t^{*}_{j \gets i}}$ indicates task ${t_j \in \mathcal{T}}$ to which agent ${a_i}$ should have joined according to the optimal partition $\Pi^*$; and 
${S_{j} \in \Pi}$ is the coalition for task ${t_j}$ whose participants follow the Nash stable partition $\Pi$. 

The right-hand side of the inequality in Equation (\ref{Eq_OPT_2}) can be rewritten as
\begin{equation}
\begin{split}
	u_i ( & t^{*}_{j \gets i} ,  | S_{j} \cup \{a_i\} | )  =  u_i ( t^{*}_{j \gets i} , | S^{*}_{j}| )  - \\
	&  \big\{ u_i ( t^{*}_{j \gets i} , | S^{*}_{j}| ) - u_i ( t^{*}_{j \gets i} , | S_{j} \cup \{a_i\} | )  \big\},  
	\label{Eq_OPT_3}
\end{split}
\end{equation}
where ${S^{*}_{j} \in \Pi^{*}}$ is the ideal coalition of task ${t^{*}_{j \gets i}}$ that {\col maximize}s the objective function.  

By summing {\colg over} all the agents, the inequality in Equation (\ref{Eq_OPT_2}) can be said that  
\begin{IEEEeqnarray}{rCl}
	\IEEEeqnarraymulticol{3}{l}{
	\sum_{\forall a_i \in \mathcal{A}} u_i (t_{{\colg \Pi(i)}},|S_{{\colg \Pi(i)}}|) 
	}\nonumber \\ \quad\quad
	& \ge & \sum_{\forall a_i \in \mathcal{A}} u_i ( t^{*}_{j \gets i} , | S^{*}_{j}| ) \nonumber \\ && - \sum_{\forall a_i \in \mathcal{A}} \{ u_i ( t^{*}_{j \gets i} , | S^{*}_{j}| ) - u_i ( t^{*}_{j \gets i} , | S_{j} \cup \{a_i\} | )  \}.  
	\nonumber\\*
	\label{Eq_OPT_4}
\end{IEEEeqnarray}
The left-hand side of the inequality in Equation (\ref{Eq_OPT_4}) represents the objective function value of the Nash stable partition $\Pi$, i.e., $J_{GRAPE}$, 
and the first term of the right-hand side is the optimal value, i.e., $J_{OPT}$. 
The second term in the right-hand side can be interpreted as the summation of the utility lost of each agent caused by the belated decision to its optimal task, provided that {\colg the} other agents still follow the Nash stable partition.

The upper bound of the second term is given by
\begin{equation}
\sum_{j=1}^{{\colg n_t}} | S^{*}_{j}| \cdot \max_{a_i \in S^{*}_j} \{  u_i ( t^{*}_{j \gets i} , | S^{*}_{j}| ) - u_i ( t^{*}_{j \gets i} , | S_{j} \cup \{a_i\} | ) \}.  
\label{Eq_OPT_5}
\end{equation}
%
%
This is at most
\begin{equation}
	\sum_{\forall S_j \in \Pi} \max_{a_i \in \mathcal{A}, p \le |\mathcal{A}|} L_{ij}[p]	\equiv \lambda, 
	\label{Eq_OPT_6}
\end{equation}
where $L_{ij}[p] = p \cdot ( u_i ( t_{j} , p ) - u_i ( t_{j} , | S_{j} \cup \{a_i\} | ) )$.

Hence, the inequality in Eqn (\ref{Eq_OPT_4}) can be rewritten as 
\begin{equation*}
J_{GRAPE} \ge J_{OPT} - \lambda.
\end{equation*}
Dividing both sides by $J_{GRAPE}$ and rearranging them yield the suboptimality lower bound of the Nash stable partition, as given by Equation (\ref{Eq_OPT_1}).
\end{proof}

Although Theorem \ref{OPT_Bound} does not provide a fixed-value lower bound, it can be determined as long as a Nash stable partition and agents' individual utility functions are given.
Nevertheless, as a special case, if the social utility for any coalition is non-decreasing (or monotonically increasing) in terms of the number of co-working agents, 
then we can obtain a fixed-value lower bound for the suboptimality of a Nash stable partition.

\begin{theorem}[\emph{Suboptimality lower bound: a special case}]
\label{optimality}
Given an instance $(\mathcal{A},\mathcal{T},\mathcal{P})$ of GRAPE, 
if 
(i) the social utility for any coalition is non-decreasing with regard to the number of participants,
i.e., for any $S_j \subseteq \mathcal{A}$ and $a_l \in \mathcal{A}\setminus S_j$, it holds that 
\begin{equation*}
\sum_{\forall a_i \in S_j} u_i(t_j, |S_j|) \le \sum_{\forall a_i \in S_j \cup \{a_l\}} u_i(t_j, |S_j \cup \{a_l\}|), 
\end{equation*} 
and 
(ii) all the individual utilities can derive SPAO preference relations,
then a Nash stable partition $\Pi$ obtained by GRAPE provide{\colg s at least} $50\%$ of suboptimality in terms of the global utility. 
\end{theorem}

\begin{proof}
Firstly, we introduce some definitions and notations that facilitate to describe this proof. 
Given a partition $\Pi$ of an instance $(\mathcal{A},\mathcal{T},\mathcal{P})$, the global utility is denoted by 
\begin{equation}\label{eqn:social_util}
\begin{split}
V(\Pi) :=  \sum_{\forall a_i \in \mathcal{A}} u_i(t_{\Pi(i)},|S_{\Pi(i)}|).
\end{split}
\end{equation}

We use operator $\oplus$ as follows. 
Given any two partitions $\Pi^{A} = \{S^A_0, ..., S^A_{n_t}\}$ and $\Pi^{B}  = \{S^B_0, ..., S^B_{n_t}\}$, 
\begin{equation*}
\Pi^{A} \oplus \Pi^{B} := \{S_0^A \cup S_0^B, \ S_1^A \cup S_1^B, ..., \ S_{n_t}^A \cup S_{n_t}^B  \}.
\end{equation*}
{\colg Since} $\cup^{n_t}_{j=0}S^A_j=\cup^{n_t}_{j=0}S^B_j=\mathcal{A}${\colg ,}
there may exist the same agent $a_i$ even in two different coalitions in $\Pi^{A} \oplus \Pi^{B}$. 
For instance, suppose that $\Pi^{A} = \{ \{a_1\}, \{a_2\}, \{a_3\}\}$ and $\Pi^{B} = \{ \emptyset, \{a_1, a_3\}, \{a_2\}\}$.  
Then, $\Pi^{A} \oplus \Pi^{B} = \{ \{a_1\}, \{a_1, a_2, a_3\}, \{a_2, a_3\}\}$. 
We regard such an agent as two different agents in $\Pi^{A} \oplus \Pi^{B}$. 
Accordingly, the operation may increase the number of total agents in the resultant partition. 

Using the definitions described above, condition (i) implies that
\begin{equation}\label{eqn:socialutil_nondec}
V(\Pi^A) \le V(\Pi^A \oplus \Pi^B).
\end{equation}

From now on, we will show that $\frac{1}{2} V(\Pi^*) \le V(\hat{\Pi})$, 
where $\Pi^{*} = \{S^*_0, S^*_1, ...,S^*_{n_t}\}$ is an optimal partition and 
$\hat{\Pi} = \{\hat{S}_0, \hat{S}_1, ...,\hat{S}_{n_t}\}$ is a Nash stable partition. 
By doing so, this theorem can be proved. 
From the definition in Equation (\ref{eqn:social_util}), it can be said that
\begin{equation}\label{eqn:eqn_15}
\begin{split}
V(\hat{\Pi} \oplus \Pi^*)= &  \sum_{\forall {a}_i \in {\mathcal{A}}} u_i(t_{\hat{\Pi}(i)}, |\hat{S}_{\hat{\Pi}(i)} \cup {S}^*_{\hat{\Pi}(i)}|) \\
& + \sum_{\forall {a}_i \in {\mathcal{A}}^{-}} u_i(t_{\Pi^*(i)}, |\hat{S}_{\Pi^*(i)} \cup {S}^*_{\Pi^*(i)}|),
\end{split}
\end{equation}
where $\mathcal{A}^-$ is the set of agents whose decisions follow not the Nash stable partition $\hat{\Pi}$ but only the optimal partition $\Pi^*$. 
Due to condition (ii), 
the first term of the right-hand side in Equation (\ref{eqn:eqn_15}) is no more than 
\begin{equation}
\sum_{\forall {a}_i \in {\mathcal{A}}} u_i(t_{\hat{\Pi}(i)}, |\hat{S}_{\hat{\Pi}(i)}|)  \equiv V(\hat{\Pi}). 
\end{equation}
Likewise, the second term is also at most
\begin{equation}
\sum_{\forall {a}_i \in {\mathcal{A}}^{-}} u_i(t_{\Pi^*(i)}, |\hat{S}_{{\Pi}^*(i)} \cup \{a_i\}|). 
\end{equation}
By the definition of Nash stability (i.e., for every agent $a_i \in \mathcal{A}$, 
$u_i(t_{\hat{\Pi}(i)}, |\hat{S}_{\hat{\Pi}(i)}|) \ge  u_i(t_j, |\hat{S}_j \cup \{a_i\}|)$, $\forall {\hat{S}_j \in \hat{\Pi}}$),  
the above equation is at most
\begin{equation}
\sum_{\forall {a}_i \in {\mathcal{A}}^{-}} u_i(t_{\hat{\Pi}(i)}, |\hat{S}_{\hat{\Pi}(i)}|),
\end{equation}
which is also no more than, because of $\mathcal{A}^- \subseteq \mathcal{A}$, 
\begin{equation}
\sum_{\forall {a}_i \in {\mathcal{A}}} u_i(t_{\hat{\Pi}(i)}, |\hat{S}_{\hat{\Pi}(i)}|) \equiv V(\hat{\Pi}). 
\end{equation}

{\colg Accordingly}, the left-hand side of Equation (\ref{eqn:eqn_15}) holds the following inequality:
\begin{equation}
V(\hat{\Pi} \oplus \Pi^*) \le 2 V(\hat{\Pi}).
\end{equation}
Thanks to Equation (\ref{eqn:socialutil_nondec}), it follows that
\begin{equation*}
V(\Pi^*) \le V(\hat{\Pi} \oplus \Pi^*).
\end{equation*}
Therefore, $V(\Pi^*) \le 2 V(\hat{\Pi})$, which completes this proof. 

\end{proof}

\subsection{Adaptability} 
Our proposed framework is also adaptable {\colg to} dynamic environments {\colg such as unexpected addition or loss of agents or tasks}, owing to its fast convergence to a Nash stable partition. 
Thanks to Lemma \ref{lemma_1}, if a new agent additionally joins an ongoing mission in which an  assignment was already determined, the number of iterations required for converging to a new Nash stable partition is at most the number of the total agents. 
Responding to any environmental change, the framework is able to establish a new agreed task assignment within polynomial time. 

\subsection{{\col Robustness in Asynchronous Environments}}

{\col
In the proposed framework, for every iteration, each agent does not need to wait until nor ensure that its locally-known information has been propagated to a certain neighbor group.
Instead, as described in Remark \ref{remark:comm_overhead}, it is enough for the agent to receive the local information from one of its neighbors, to make a decision, and to send the updated partition back to some of its neighbors. 
Temporary disconnection or non-operation of some agents may cause dummy iterations additionally.  
However, it does not affect the existence of, the convergence toward, and the suboptimality of a Nash stable partition under the proposed framework, which is also supported by Section \ref{sec:result_robust}.
}

\section{GRAPE with Minimum Requirements}\label{sec:min_rqmt}

This section addresses another task allocation problem where each task may require at least a certain number of agents for {\colg its} completion. 
This problem can be defined as follows. 
\begin{problem}\label{prob_minrqmt}
Given a set of agents $\mathcal{A}$ and a set of tasks $\mathcal{T}$, 
the objective is to find an assignment such that 
\begin{equation}\label{eqn:obj_minrqmt}
\max_{\{x_{ij}\}} \sum_{\forall a_i \in \mathcal{A}} \sum_{\forall t_j \in \mathcal{T}} u_{i}(t_j, p) x_{ij}{\colg ,}
\end{equation}
subject to
\begin{equation}\label{eqn:prob_minrqmt_min_rqmt}
\sum_{\forall a_i \in \mathcal{A}} x_{ij} \ge R_j{\colg ,} \quad \forall t_j \in \mathcal{T}{\colg ,}
\end{equation}
\begin{equation}
\sum_{\forall t_j \in \mathcal{T}} x_{ij} \le 1{\colg ,} \quad \forall a_i \in \mathcal{A}{\colg ,}
\end{equation}
\begin{equation}
x_{ij} \in \{0,1\}{\colg ,} \quad \forall (a_i,t_j) \in \mathcal{A} \times \mathcal{T}{\colg ,}
\end{equation}
where $R_j \in {\colg\mathbb{N} \cup \{0\}}$ is the number of minimum required agents for task $t_j$, and all the other variables are identically defined as those in Problem \ref{prob_basic}. 
Here, it is considered that, for $\forall a_i \in \mathcal{A}$ and $\forall t_j \in \mathcal{T}$,
\begin{equation}\label{eqn:no_reward}
u_i (t_j, p) = 0  \quad \text{if $p < R_j$}
\end{equation}
because task $t_j$ cannot be completed in this case.
{\colg Note that any task $t_j$ without such a requirement is regarded to have $R_j = 0$.}
\end{problem}

{\colg For each task $t_j$ having $R_j > 0$,}
even if $u_i(t_j,p)$ is monotonically decreasing at $p \ge R_j$, 
the individual utility can not be simply transformed to a preference relation holding SPAO because of {\colg Equation} (\ref{eqn:no_reward}). 
Thus, we need to modify the utility function to yield alternative values for the case when $p < R_j$. 
We refer to the modified utility as \emph{auxiliary individual utility} $\tilde{u}_i$, which is defined as
\begin{equation}\label{eqn:aux_util}
\tilde{u}_{i}(t_j, p) =  \begin{cases}
{u}^0_{i}(t_j,p) & \text{if $p \le R_j$} \\
{u}_{i}(t_j, p) & \text{otherwise,} 
\end{cases}
\end{equation}
where ${u}^0_{i}(t_j,p)$ is the \emph{dummy utility} of agent $a_i$ with regard to task $t_j$ when $p \le R_j$.

The dummy utility is intentionally used also for the case when $p = R_j$ in order to find an assignment that holds Equation (\ref{eqn:prob_minrqmt_min_rqmt}). 
For this, the auxiliary individual utility should satisfy the following condition.
\begin{condition}\label{con:min_rqmt}
For every agent $a_i \in \mathcal{A}$, its preference relation $\mathcal{P}_i$ holds that, 
for any two tasks $t_j, t_k \in \mathcal{T}$, 
\begin{equation*}
(t_j,R_j) \succ_i (t_k,R_k+1). 
\end{equation*}
This condition enables every agent to prefer a task for which the number of co-working agents is less than its minimum requirement, over any other tasks whose requirements are already fulfilled. 
Under this condition, as long as the agent set $\mathcal{A}$ is such that $|\mathcal{A}| \ge \sum_{\forall t_j \in \mathcal{T}} R_j$ 
and a Nash stable partition is found, 
the resultant assignment satisfies Equation (\ref{eqn:prob_minrqmt_min_rqmt}). 
\end{condition}

\begin{proposition}\label{prop:min_rqmt}
Given an instance of Problem \ref{prob_minrqmt} where $u_i(t_j, p)$ $\forall i$ $\forall j$ is a monotonically decreasing function with regard to $\forall p \ge R_j$, 
if the dummy utilities ${u}^0_{i}(t_j,p)$ $\forall i$ $\forall j$ in (\ref{eqn:aux_util}) are set to satisfy Condition \ref{con:min_rqmt} and SPAO for $\forall p \le R_j$, 
then all the resultant auxiliary individual utilities $\tilde{u}_i(t_j,p)$ $\forall i$ $\forall j$ $\forall p$ can be transformed to a $n_a$-tuple of preference relations $\mathcal{P}$ that hold Condition \ref{con:min_rqmt} as well as SPAO for $\forall p \in \{1,...,n_a\}$. 
In the corresponding instance of GRAPE $(\mathcal{A},\mathcal{T}, \mathcal{P})$, 
a Nash stable partition can be determined within polynomial times as shown in Theorems \ref{NASH} and \ref{Nash_poly} because of SPAO, and the resultant partition can satisfy Equation (\ref{eqn:prob_minrqmt_min_rqmt}) due to Condition \ref{con:min_rqmt}. 
\end{proposition}

Let us give an example. 
Suppose that there exist 100 agents $\mathcal{A}$, and 3 tasks $\mathcal{T} = \{t_1, t_2, t_3\}$ where only $t_3$ has its minimum requirement $R_3 = 5$; 
for every agent $a_i \in \mathcal{A}$, individual utilities for $t_1$ and $t_2$, i.e., $u_i(t_1,p)$ and $u_i({\colb t_2},p)$, are much higher than that for $t_3$ in $\forall p \in \{1,...,100\}$.
We can find a Nash stable partition for this example, as described in Proposition \ref{prop:min_rqmt},
by setting 
${u}^0_{i}(t_j,p) {\colg = \max_{\forall t_j}\{{u}_{i}(t_j,R_j + 1)\} + \beta}$ for $\forall p \le R_j$, {\colg $\forall a_i \in \mathcal{A}$, 
where $\beta > 0$ is an arbitrary positive constant}.

After a Nash stable partition is found, in order to compute the objective function value in (\ref{eqn:obj_minrqmt}), 
the original individual utility function $u_i$ should be used instead of the auxiliary one $\tilde{u}_i$.

\begin{proposition}\label{prop:subbound_minrqmt}
{\colg Given} a Nash stable partition $\Pi$ obtained by implementing Proposition \ref{prop:min_rqmt},  
its suboptimality bound $\alpha$ is such that
\begin{equation}\label{eqn:bound_minrqmt}
\alpha \ge \frac{J_{GRAPE}}{J_{GRAPE} + \tilde{\lambda}} \cdot \frac{J_{GRAPE}}{J_{GRAPE} + \delta}. 
\end{equation}
Here, $\delta \equiv \tilde{J}_{GRAPE} - J_{GRAPE}$,  
where $\tilde{J}_{GRAPE}$ {\colg (or} ${J}_{GRAPE}${\colg )} is the objective function value in (\ref{eqn:obj_minrqmt}) using $\tilde{u}_i$ {\colg (or using} ${u}_i${\colg )} given the Nash stable partition.
Likewise, $\tilde{\lambda}$ is the value in (\ref{eqn:lambda}) using $\tilde{u}_i$.
In addition to this, if every $\tilde{u}_i$ satisfies the conditions for Theorem \ref{optimality}, then 
\begin{equation}\label{eqn:bound2_minrqmt}
\alpha \ge \frac{1}{2} \cdot \frac{J_{GRAPE}}{J_{GRAPE} + \delta}.
\end{equation}
\end{proposition}

\begin{proof}
Since the Nash stable partition $\Pi$ is obtained by using $\tilde{u}_i$, it can be said from Equations (\ref{eqn:approx_ratio}) and (\ref{Eq_OPT_1}) that
\begin{equation}\label{eqn:proposition_2}
\frac{\tilde{J}_{GRAPE}}{\tilde{J}_{OPT}} \ge \frac{\tilde{J}_{GRAPE}}{\tilde{J}_{GRAPE} + \tilde{\lambda}}.
\end{equation}
Due to the fact that $\tilde{u}_i(t_j,p) \ge {u}_i(t_j,p)$ for $\forall i,j,p$, 
it is clear that $\tilde{J}_{GRAPE} \ge {J}_{GRAPE}$ and $\tilde{J}_{OPT} \ge {J}_{OPT}$. 
By letting that $\delta := \tilde{J}_{GRAPE} - J_{GRAPE}$, the left term in (\ref{eqn:proposition_2}) is at most 
$({J}_{GRAPE} + \delta)/{{J}_{OPT}}$. 
Besides, the right term in (\ref{eqn:proposition_2}) is a monotonically-increasing function with regard to $\tilde{J}_{GRAPE}$, and thus, 
it is lower bounded by ${J}_{GRAPE}/(J_{GRAPE} \ {\colg +} \ \tilde{\lambda})$.
From this, Equation (\ref{eqn:proposition_2}) can be rewritten as Equation $(\ref{eqn:bound_minrqmt})$ by multiplying $J_{GRAPE}/(J_{GRAPE}+\delta)$.

Likewise, for the case when every $\tilde{u}_i$ satisfies the conditions for Theorem \ref{optimality}, it can be said that
$\tilde{J}_{GRAPE} \ge 1/2 \cdot \tilde{J}_{OPT}$, which can be transformed into Equation (\ref{eqn:bound2_minrqmt}) as shown above.
\end{proof}

Notice that if $\delta = 0$ for the Nash stable partition in Proposition \ref{prop:subbound_minrqmt}, then the suboptimality bounds become equivalent to those in Theorems \ref{OPT_Bound} and \ref{optimality}.

\section{Simulation and Results}\label{Results}

This section validates the performances of the proposed framework with respect to its scalability, suboptimality, adaptability against dynamic environments, and robustness in asynchronous environments.  


\subsection{Mission Scenario and Settings}\label{sec:setting}


\subsubsection{Utility functions}
Firstly, we introduce the social and individual utilities used in this numerical experiment. 
%
We consider that if multiple robots execute a task together as a coalition, {\colg then} they are given a certain level of reward for the task. 
The amount of the reward varies depending on the number of the co-working agents. 
The reward is shared with the agents, 
and each agent's individual utility is considered as the shared reward minus the cost required to personally spend {\colg on} the task (e.g., fuel consumption for movement). 
In this experiment, \emph{the equal fair allocation rule} \cite{Saad2009, Saad2011} is adopted.
Under the rule, a task's reward is equally shared {\col among} the members. 
Therefore, the individual utility of agent $a_i$ executing task $t_j$ with coalition $S_j$ is defined as
\begin{equation}
\label{basic_utility_ftn}
u_i(t_j,|S_j|)=r(t_j, |S_j|)/|S_j|-c_i(t_j),
\end{equation}
where $r(t_j, |S_j|)$ is the reward from task ${t_j}$ when it is executed by $S_j$ together, 
and $c_i(t_j)$ is the cost that agent ${a_i}$ needs to pay for the task.
Here, we simply set the cost as a function of the distance from agent $a_i$ to task $t_j$. 
We set that if $u_i(t_j,|S_j|)$ is not positive, agent $a_i$ prefers to join $S_{\phi}$ over $S_j$. 

This experiment considers two types of tasks. 
For the first type, a task's reward becomes higher as the number of participants gets close to a specific desired number. 
We refer to such a task as a \emph{peaked-reward} task, and its reward can be defined as 
\begin{equation}
\label{peaked_task}
r(t_j, |S_j|)=\frac{r^{\max}_{j} \cdot |S_j|}{n^d_j} \cdot e^{-|S_j|/n^d_j + 1},
\end{equation}
where ${n^d_j}$ represents the desired number, and ${r^{\max}_j}$ is the peaked reward in case that $n^d_j$ of agents are involved in. 
Consequently, the individual utility {\colg of} agent $a_i$ with regard to task $t_j$ becomes the following equation:
\begin{equation}
\label{utility_ftn_1}
u_i(t_j,|S_j|)=\frac{r^{\max}_j}{n^d_j} \cdot e^{-|S_j|/n^d_j + 1}- c_i(t_j).
\end{equation}

For the second type, a task's reward becomes higher as more agents are involved, but the corresponding marginal gain decreases. 
This type of tasks is said to be \emph{submodular-reward}, and the reward can be defined as  
\begin{equation}
\label{submodular_task}
r(t_j, |S_j|)=r^{\min}_j \cdot \log_{{\epsilon}_j}(|S_j|+{\epsilon}_j-1),
\end{equation}
where ${r^{\min}_j}$ indicates the reward obtained if there is only one agent {\colg is} involved, and ${{\epsilon}_j} > 1$ is the design parameter regarding the diminishing marginal gain. 
The resultant individual utility becomes as follows:
\begin{equation}
\label{utility_ftn_2}
u_i(t_j,|S_j|)=r^{\min}_j \cdot \log_{{\epsilon}_j}(|S_j|+{\epsilon}_j-1)/|S_j| - c_i(t_j).
\end{equation}

\begin{figure}[t]
\centering
\subfloat[Social utility]{\includegraphics[width=0.50\linewidth]{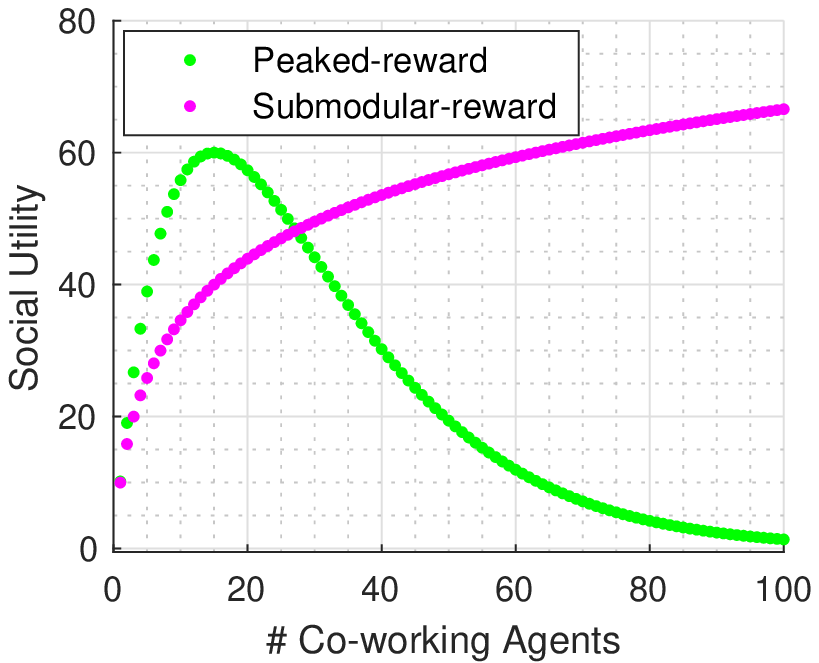}} 
\subfloat[Individual utility]{\includegraphics[width=0.50\linewidth]{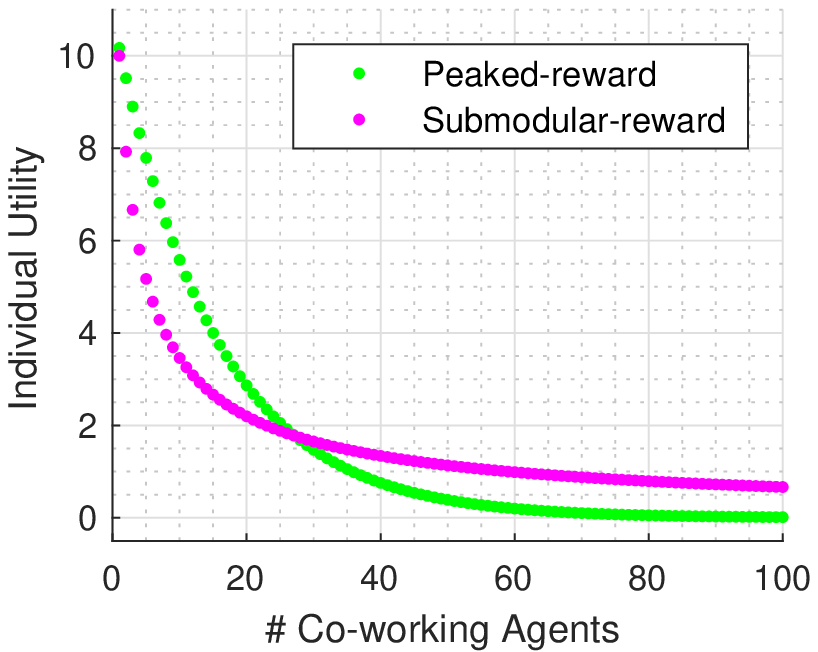}}
\caption{\col Examples of utility functions used in the numerical experiment are shown, depending on the two different task types (i.e., peaked-reward and submodular-reward): \textbf{(a)} the social utility of a coalition; \textbf{(b)} an agent's individual utility. 
}
\label{fig_utility_ftn}
\end{figure}

Figure \ref{fig_utility_ftn} illustrates examples of the social utilities and individual utilities for the task types introduced above.  
For simplification, agents' costs are ignored in the figure. 
We set ${r^{\max}_j}$, ${n^d_j}$, ${r^{\min}_j}$ and ${\epsilon_j}$ to be 60, 15, 10, and 2, respectively.
Notice that the individual utilities are monotonically decreasing in both cases, as depicted in Figure \ref{fig_utility_ftn}(b).
Therefore, given a mission that entails these task types, we can generate an instance ${(\mathcal{A},\mathcal{T},\mathcal{P})}$ of GRAPE that holds SPAO. 





\subsubsection{Parameters generation}
In the following sections, we will mainly {\col utilize} Monte Carlo simulations. 
At each run, $n_t$ tasks and $n_a$ agents are uniform-randomly located in a $1000 \ m \times 1000 \ m$ arena and a $250 \ m \times 250 \ m$ arena within there, respectively. 
For a scenario including peaked-reward tasks, 
$r^{\max}_j$ is randomly generated from a uniform distribution over $[1000 , 2000] \times n_a/n_t${\colg ,} 
and $n^d_j$ is {\colg set to be} the rounded value of $(r^{\max}_j/{\sum_{\forall t_k \in \mathcal{T^*}} r^{\max}_k}) \times n_a$. 
For a scenario including submodular-reward tasks, 
$\epsilon_j$ is set as 2, and $r^{\min}_j$ is uniform-randomly generated over $[1000, 2000] \times 1/\log_{{\epsilon}_j}{(n_a/n_t + 1)}$. 


\subsubsection{Communication network}
Given a set of agents, 
their communication network is strongly-connected  
{\colg in a way that only contains a bidirectional minimum spanning tree with consideration of the agents' positions.} 
Furthermore, we also consider the fully-connected network in some experiments in order to examine the influence of the network.
The communication network is randomly generated at each instance, and is assumed to be sustained during a mission except the robustness test simulations in Section \ref{sec:result_robust}.

\subsection{Scalability}
\label{sec:result_scalability}
To investigate the effectiveness of $n_t$ and $n_a$ upon the scalability of the proposed approach, 
we conduct Monte Carlo simulation{\colg s} with 100 runs
for the scenarios introduced in Section \ref{sec:setting} with a fixed $n_t = 20$ and various $n_a \in \{80, 160, 240, 320\}$ and for those with $n_a = 160$ and $n_t \in \{5, 10, 15, 20\}$.
Figure \ref{fig_result_scalability} shows the statistical results using box-and-whisker plots, where  
the green boxes indicate the results from the scenarios with the peaked-reward tasks and the magenta boxes are those with the submodular-reward tasks. 
The blue and red lines connecting the boxes represent the average value for each test case $(n_a, n_t)$ under a strongly-connected network and the fully-connected network, respectively.

\begin{figure}[t]
\centering
\subfloat[$n_a \in \{80, 160, 240, 320\}$ with fixed $n_t = 20$]{ 
	\includegraphics[width=0.50\linewidth]{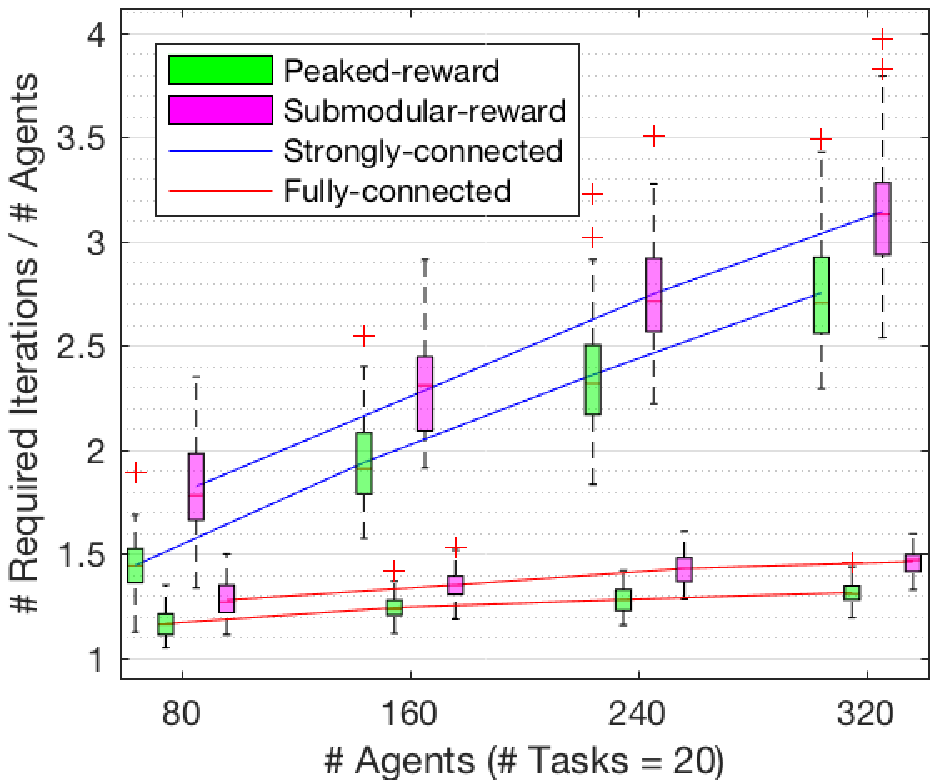} 
	\includegraphics[width=0.50\linewidth]{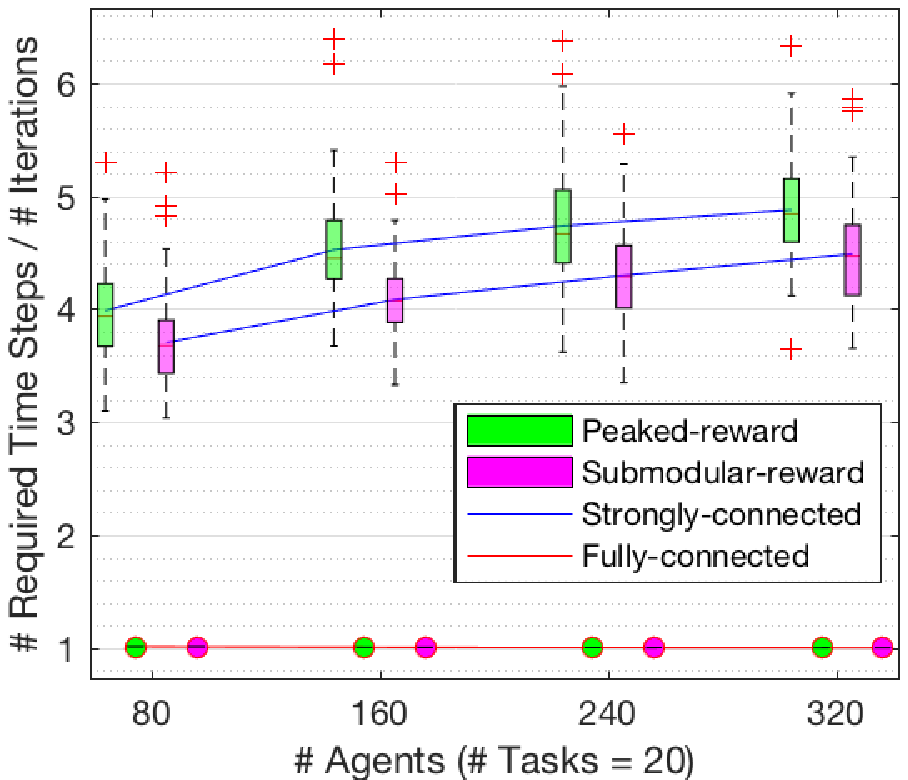}	
	}\hfill
\subfloat[$n_t \in \{5,10,15,20\}$ with fixed $n_a = 160$]{ %
	\includegraphics[width=0.50\linewidth]{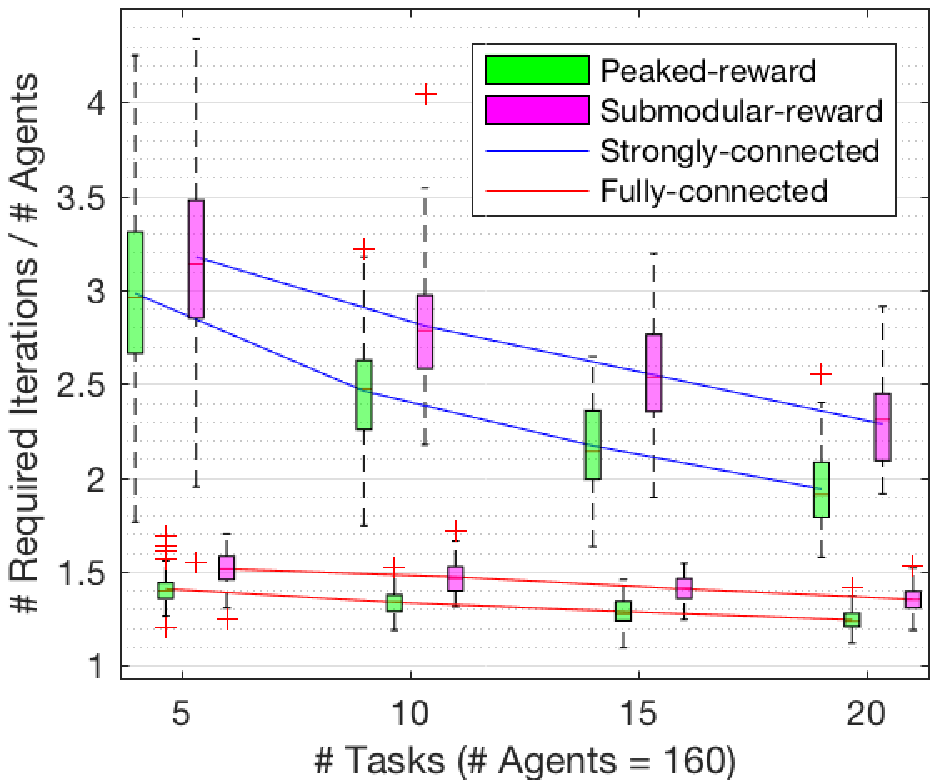} 	
	\includegraphics[width=0.50\linewidth]{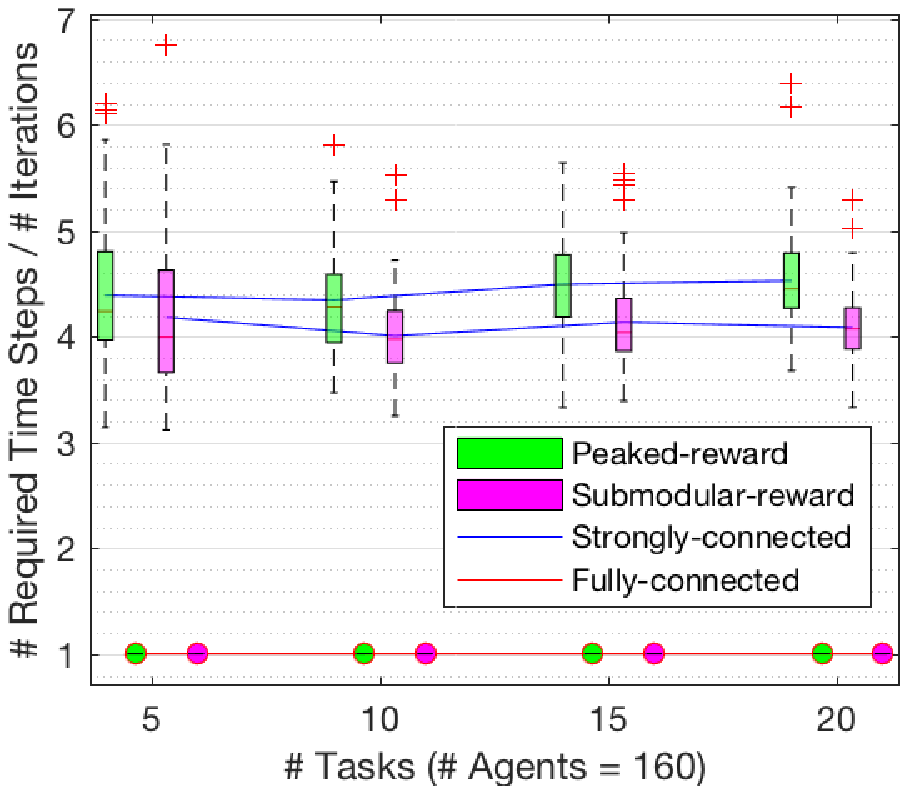}
	}\hfill
\caption{{\col Convergence performance of the proposed framework is shown,} depending on communication networks (i.e., Strongly-connected vs. Fully-connected) and utility function types (i.e., Peaked-reward vs. Submodular-reward) {\col with different number of agents and tasks: \textbf{(Left)} the number of (normal) iterations happened relative to that of agents; \textbf{(Right)} the number of time steps happened (i.e., normal and dummy iterations) relative to that of iterations. }}
\label{fig_result_scalability}
\end{figure}

%

The left subfigure in Figure \ref{fig_result_scalability}(a) shows that the ratio of the number of required (normal) iterations to that of agents linearly increases as more agents are involved. 
This implies that {\colg the proposed framework has quadratic complexity with regard to the number of agents} (i.e., $C_1 n_a^2$), as stated in Theorem \ref{Nash_poly}, but with $C_1$ being much less than $\frac{1}{2}$, which is the value from the theorem.
$C_1$ can become {\colg even lower} (e.g., $C_1 = 5 \times 10^{-4}$ in the experiments) under the fully-connected network. 
Such $C_1$ being smaller than $\frac{1}{2}$ may be explained by Remark \ref{remark:num_iter_practice}: 
the algorithmic efficiency of Algorithm \ref{algorithm} can reduce unnecessary iterations that may be induced in the procedure of the proof for Theorem \ref{Nash_poly}.  

On the other hand, the left subfigure in Figure \ref{fig_result_scalability}(b) {\colg shows} that the number of required iterations decreases with regard to the number of tasks.
This {\colg trend} may be caused by the fact that more selectable options provided to {\colg the} fixed number of agents can reduce possible conflicts between the agents. 

Furthermore, in {\colg the two results}, the trends regarding either $n_a$ or $n_t$ have higher slope{\colg s} under a strongly-connected network than {\colg those} under the fully-connected network. 
This is because the former condition is more sensitive to conflicts between agents, and thus causes additional iterations. 
For example, agents at the middle nodes of the network may change their decisions (and thus increase the number of iterations) while the local partition information of the agent at one end node is being propagated to the other end nodes. 
Such unnecessary iterations {\colg in the middle} might not have occurred if the agents at {\colg all} the end nodes were directly connected {\colg to each other}.

The right subfigures in Figure \ref{fig_result_scalability}(a) and (b) indicate that
approximately 3--4 times of dummy iterations, compared with the required number of normal iterations, are additionally needed under a strongly-connected network. 
Noting that the mean values of the graph diameter $d_G$ for the instances with $n_a \in \{80, 160, 240, 320\}$ are $36, 58, 75$ and $92$, respectively, 
the results show that the amount of dummy iterations happened is much less than the bound value, which is $d_G$ as pointed out in Section \ref{sec:algorithm_complexity}.
On the contrary, under the fully-connected network{\colg ,} there is no need of such a dummy iteration, and thus the required number of iterations and that of time steps are the same.

\subsection{Suboptimality}

This section examines the suboptimality of the proposed framework by using Monte Carlo simulations with 100 instances. 
In each instance, there are $n_t = 3$ of tasks and $n_a = 12$ of agents {\colg who} are strongly-connected. 
Figure \ref{fig_result_minimum_bound} presents the true {\colg sub}optimality {\colg of} each instance, which is the ratio of the global utility obtained by the proposed framework to that by a brute-force search, i.e., $J_{GRAPE}/J_{OPT}$, and the lower bound given by Theorem \ref{OPT_Bound}.
A blue circle and a red cross in the figure indicate the true suboptimality and  the lower bound, respectively. 
The results show that the framework provides near-optimal solutions in almost all cases and the suboptimality of each Nash stable partition is enclosed by the corresponding lower bound.  

The suboptimality may be improved 
if the agents are allowed to investigate a larger search space, for {\col example}, possible coalitions caused by co-deviation of multiple agents. 
However, this strategy in return may increase communication transactions {\colg between} the agents because they have to notice each other's willingness unless their individual utility functions are known {\colg to each} other, which is in {\colg contradiction} to Assumption \ref{assum:agents_util}. 
Besides, the computational overhead for {\colg each} agent {\colg per} iteration also becomes more expensive than $O(n_t)$, which is the complexity for {\colg unilateral} searching{\colg, as shown in Section \ref{sec:algorithm_complexity}}. 
Hence, the resultant algorithm{\colg's} complexity may hinder its {\colg practical} applicability to a large-scale multiple agent system. 

Figure \ref{fig_result_minimum_bound_large} depicts the suboptimality lower bounds for the large-size problems that were previously addressed in Section \ref{sec:result_scalability}. 
It is clearly shown that the agent communication network does not make any effect on the suboptimality lower bound of a Nash stable partition. 
Although there is no universal trend of the suboptimality with regard to $n_a$ and $n_t$ in both utility types, 
it is suggested that the features of the lower bound given by Theorem \ref{OPT_Bound} can be influenced by the utility functions considered. 
In the experiments, the suboptimality bound averagely remain above than 60--70 \%.




\begin{figure}[t]
\centering
\subfloat[Peaked-reward task]{\includegraphics[width=0.48\linewidth]{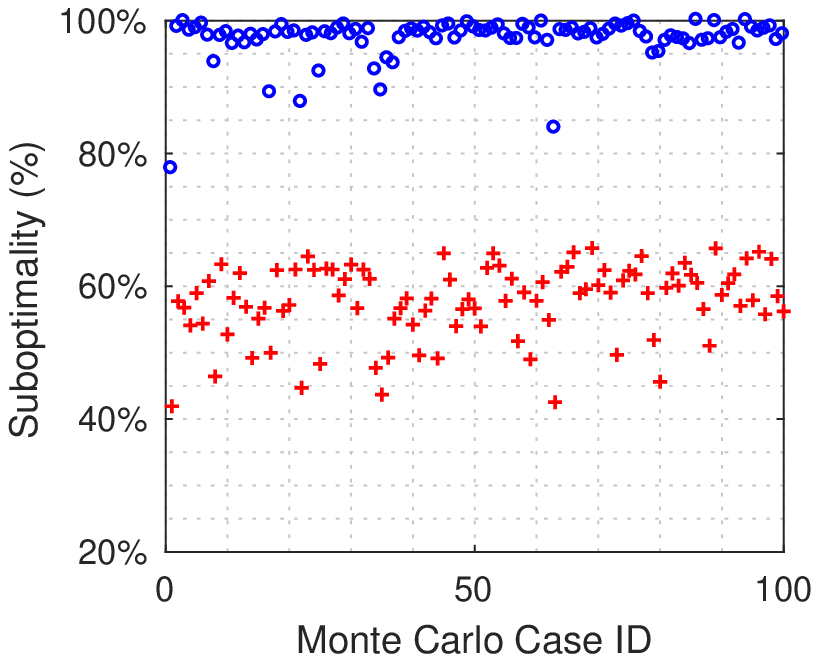} 
\label{fig_result_minimum_bound_peaked}}
\hfil
\subfloat[Submodular-reward task]{\includegraphics[width=0.48\linewidth]{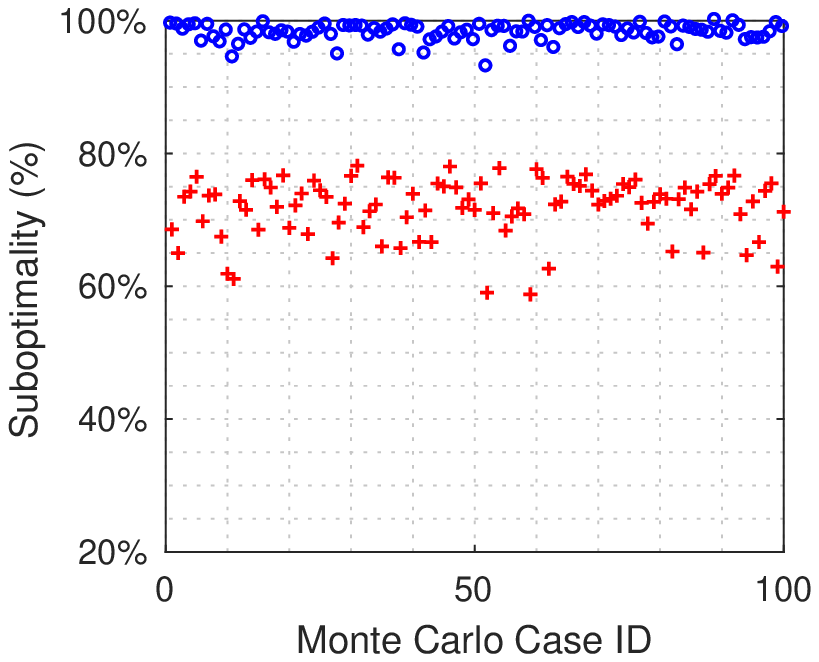}
\label{fig_result_minimum_bound_log}}
\caption{True suboptimality {\col of a Nash stable partition} obtained by GRAPE {\col for each run of the Monte Carlo simulation} ({\col denoted by a} blue circle) and its lower bound provided by Theorem \ref{OPT_Bound} ({\col denoted by a} red cross) under a strongly-connected communication network{\col : \textbf{(a)} the scenarios with peaked-reward tasks; \textbf{(b)} the scenarios with submodular-reward tasks}}
\label{fig_result_minimum_bound}
\end{figure}

\begin{figure}[t]
\centering
\subfloat[]{ 
	\includegraphics[width=0.50\linewidth]{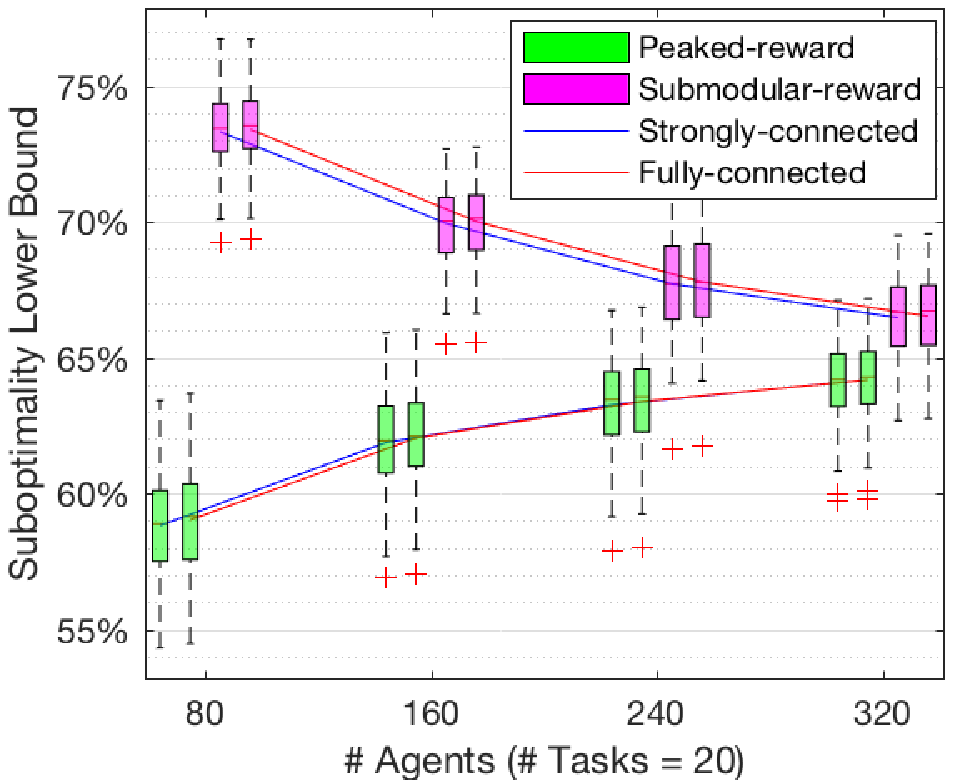}}
\subfloat[]{ %
	\includegraphics[width=0.50\linewidth]{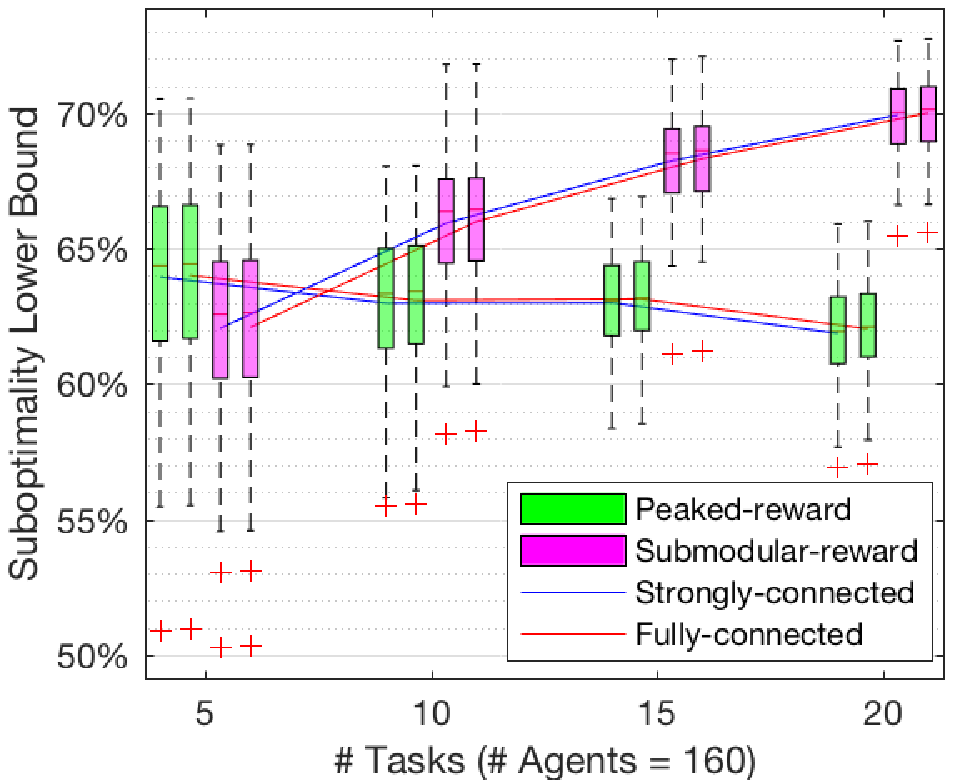}
	}
\caption{The suboptimality lower bound, given by Theorem \ref{OPT_Bound}, of a Nash stable partition {\col obtained by GRAPE, depending on communication networks (i.e., Strongly-connected vs. Fully-connected) and utility function types (i.e., Peaked-reward vs. Submodular-reward)}: 
{\col \textbf{(a)}} fixed $n_t = 20$ with varying $n_a \in \{80, 160, 240, 320\}$; {\col \textbf{(b)}} fixed $n_a = 160$ with varying $n_t \in \{5,10,15,20\}$
}
\label{fig_result_minimum_bound_large}
\end{figure}

%
%
%
%


\subsection{Adaptability}

\begin{figure}[t]
\centering
\subfloat[Dynamic Agents]{\includegraphics[width=0.9\linewidth]{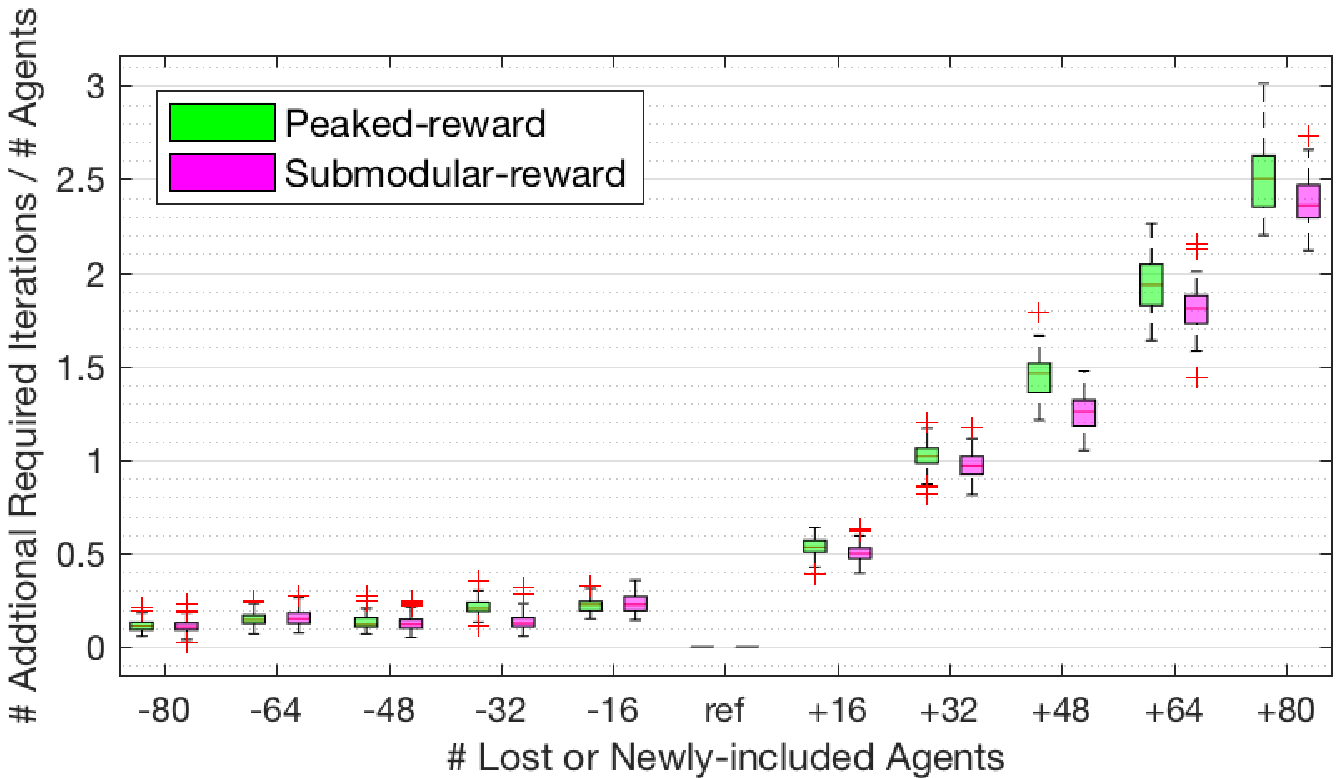}}
\hfil
\subfloat[Dynamic Tasks]{\includegraphics[width=0.9\linewidth]{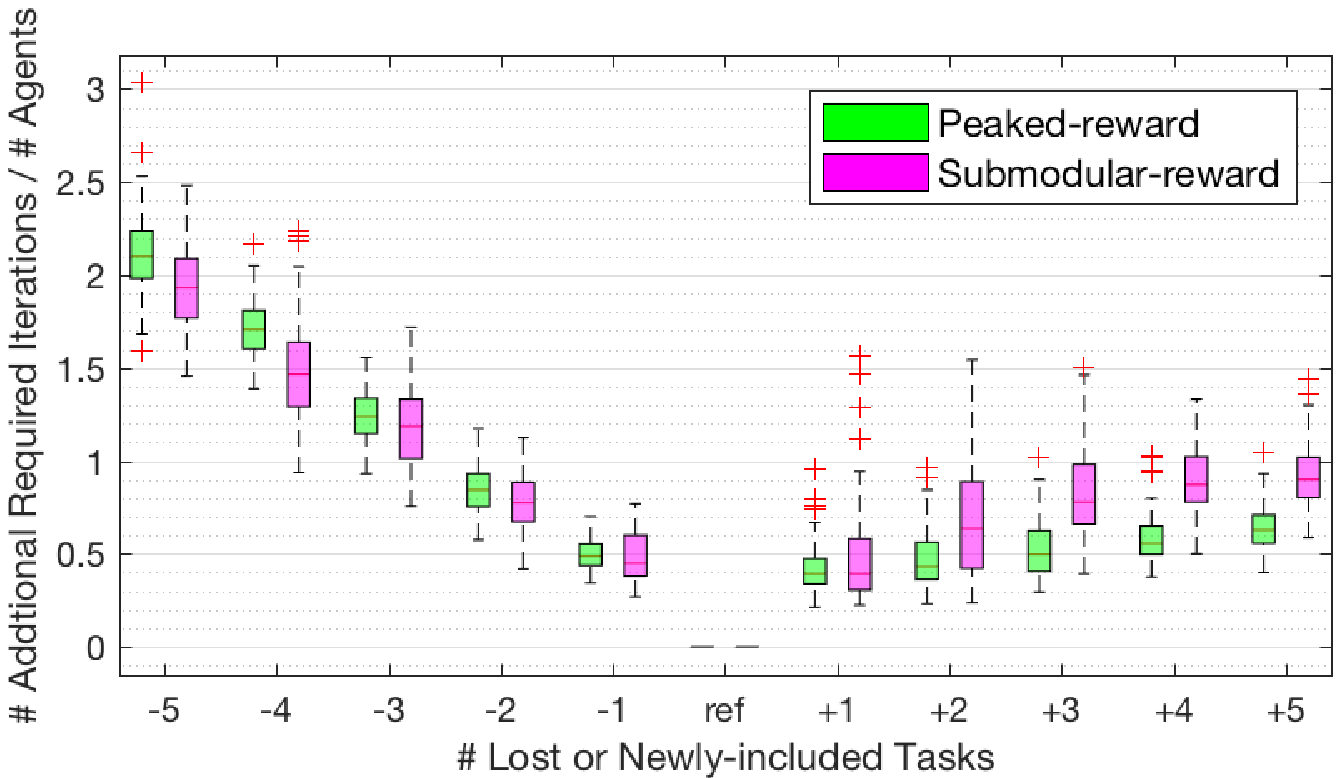}}
\caption{{\col The plot shows} the number of additional iterations required for re-converging {\col toward} a Nash stable partition relative to the number of agents {\col in the case when some agents or tasks are partially lost or newly involved (Baseline: $n_t = 10$, $n_a = 160$, and a Nash stable partition was already found)}. Negative values in the x-axis indicate that the corresponding number of existing agents or tasks are lost{\col . P}ositive values indicate that the corresponding number of new agents or tasks are included in an ongoing mission{\col . A} strongly-connected communication network is used.}
\label{fig_result_adaptability}
\end{figure}

This section discusses the adaptability of our proposed framework in response to dynamic environments such as {\colg unexpected} inclusion or loss of agents {\colg or} tasks. 
Suppose that there are 10 tasks and 160 agents in a mission, and a Nash stable partition was already found as a baseline.
During the mission, the number of agents (or tasks) changes; the range of the change is from losing 50\% of the existing agents (or tasks) to additionally including new ones as much as 50\% of them. 
For each dynamical environment, a Monte Carlo simulation with 100 instances {\colg is} performed by randomly including or excluding a subset of the corresponding number of agents or tasks. 
Here, we consider a strongly-connected communication network. 

Figure \ref{fig_result_adaptability}(a) illustrates that 
the more agents are involved {\colg additionally}, the more iterations {\colg are} required for {\colg re-}converging to a new Nash stable partition. 
This is because the inclusion of a new agent may lead to additional iterations at most as much as the number of {\colg the} total agents including the {\colg new} agent ({\colg as shown in} Lemma \ref{lemma_1}). 
{\colg On the contrary}, the loss of existing agents does not seem to have any apparent relation with the number of iterations.    
A possible explanation is that the exclusion of an existing agent is {\col favorable} to {\colg the} other agents due to SPAO preferences.
This stimulates only a limited number of agents who are preferred to move to the coalition where the excluded agent was. 
This feature induces fewer additional iterations to reach a new Nash stable partition, compared with the case of adding a new agent.

Figure \ref{fig_result_adaptability}(b) shows that eliminating existing tasks causes more iterations than {\colg including} new tasks. This can be explained {\colg by} the fact that removing any task releases the agents performing the task free and it results in extra iterations as much as the number of the freed agents. On the other hand, adding new tasks induces relatively fewer additional iterations because only some of {\colg the} existing agents are attracted to {\colg these} tasks.

In summary, as the ratio of the number of agents to that of tasks increases, the number of additional iterations {\colg for} converge{\colg nce to} a new Nash stable partition also increases. 
This result corresponds to the trend described in Section \ref{sec:result_scalability}, 
i.e., the left subfigures in Figure \ref{fig_result_scalability}(a) and (b). 
In all the cases of this experiment, the number of additionally induced iterations still remains {\colg at} the same order of the number of the given agents, which implies that the proposed framework provides excellent adaptability.

\subsection{Robustness in Asynchronous Environments}\label{sec:result_robust}

This section investigates the robustness of the proposed framework in asynchronous environments. 
This scenario assumes that a certain fraction of the given agents, which are randomly chosen at each time step, somehow can not execute Algorithm \ref{algorithm} and even can not communicate with other normally-working {\col neighbor} agents. 
We refer to such agents as \emph{non-operating} agents. 
Given that $n_t = 5$ and $n_a = 40$, 
the fractions of the non-operating agents are set as $\{0, 0.2, 0.4, 0.6, 0.8\}$. 
In each case, we conduct 100 instances of Monte Carlo experiments for which the submodular-reward {\colg tasks} are used.

Figure \ref{fig_result_async}(a) presents that 
the number of (normal) iterations required for converging {\colg to} a Nash stable partition remains {\colg at} the same level regardless of the fraction of the non-operating agents. 
Despite that, the required time steps increase as more agents become non-operating, as shown in Figure \ref{fig_result_async}(b). 
Note that \emph{time steps growth rate} means the ratio of the total required time steps to those for the case {\colg when all the agents operate normally.} 
These findings indicate that, due to communicational discontinuity caused by the non-operating agents, the framework may take more time to wait for these agents to operate again and then to disseminate {\colg locally-known} partition information over the entire agents. 
{\colg As such,} dummy iterations may increase in asynchronous environments, {\colg though}
the proposed framework is still able to find a Nash stable partition. 
Furthermore, the resultant Nash stable partition's {\colg suboptimality} lower bound obtained by Theorem \ref{OPT_Bound} is not affected, as presented in Figure \ref{fig_result_async}(c).

\begin{figure}[t]
\centering
\subfloat[]{\includegraphics[width=0.33\linewidth]{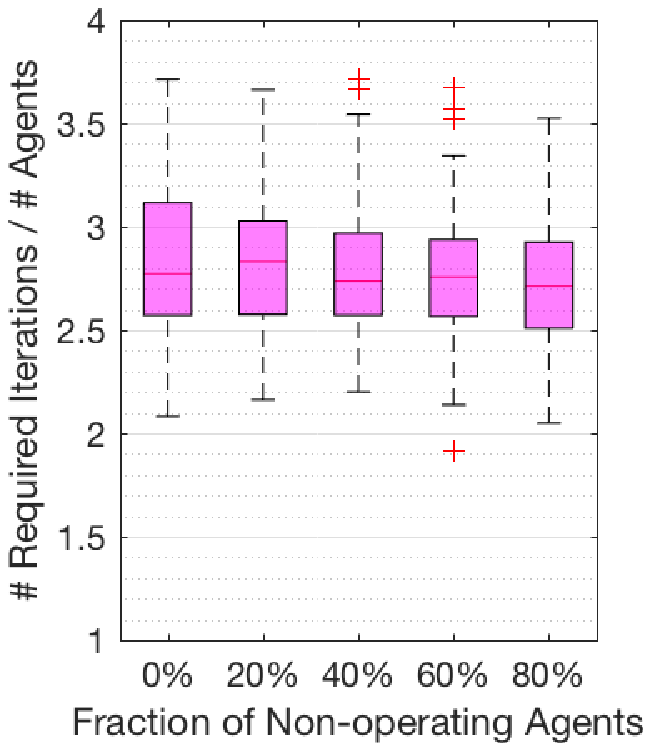}}
\hfil
\subfloat[]{\includegraphics[width=0.33\linewidth]{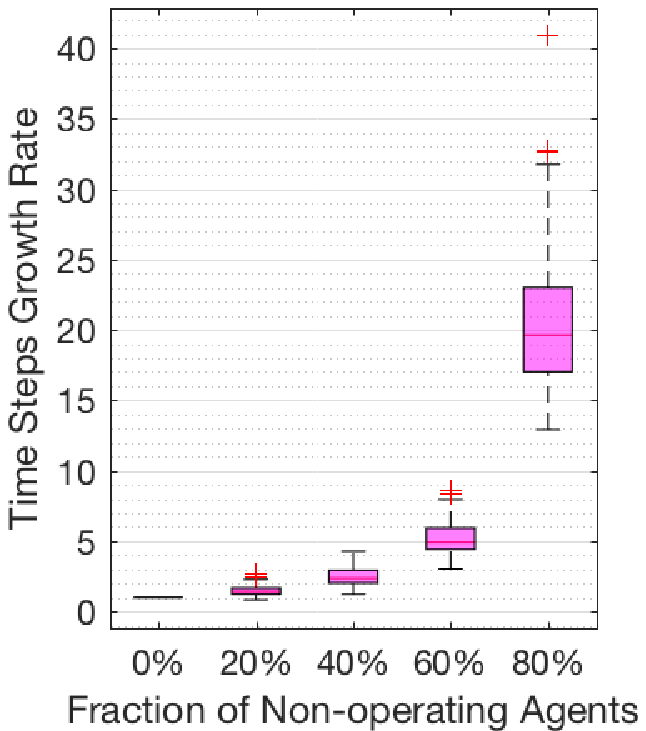}}
\hfill
\subfloat[]{\includegraphics[width=0.33\linewidth]{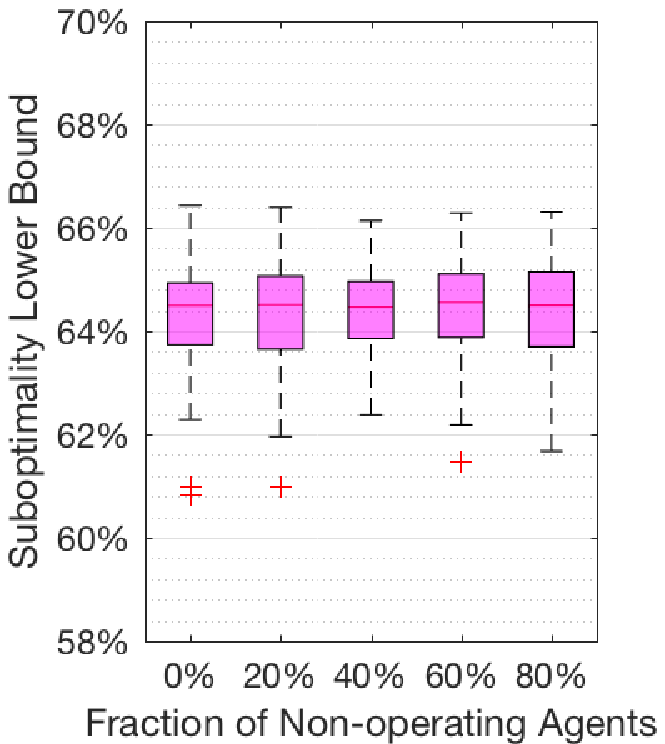}}
\caption{
Robustness test in asynchronous environments {\col at scenarios with} $n_t = 10, n_a = 160$, and the submodular-reward tasks: {\col the plot shows} the effectiveness of {\col the fraction of} non-operating agents {\col with regard to: \textbf{(a)} the number of iterations happened until convergence relative to that of agents; \textbf{(b)} the ratio of the time steps happened to those for the normal case; \textbf{(c)} the suboptimality lower bound by Theorem \ref{OPT_Bound}.}
 }
\label{fig_result_async}
\end{figure}

\subsection{Visuali{\col z}ation}

We have $n_a = 320$ agents and $n_t = 5$ tasks. 
The initial locations of the given agents are randomly generated, and the overall formation shape is different in each test scenario such as {\colg being} circle, skewed circle, and square (denoted {\colg by} Scenario \#1, \#2, and \#3, respectively). 
The tasks are also randomly located away from the agents. 
In this simulation, each agent is able to communicate with its {\colg nearby} agents {\colg within a radius of} 50 $m$. 
Here, the submodular-reward {\colg tasks} are used.

Figure \ref{fig_result_visual} shows the {\col visualize}d task allocation results{\colg ,} where the circles and the squares indicate the positions of the agents and the tasks, respectively. 
The lines between the circles represent the communication networks of the agents. 
The {\col color}ed agents are assigned to the same {\col color}ed task, for example, yellow agents belong to {\colg the} team for executing the yellow task. 
The {\colg size of a square} indicates the reward of the {\colg corresponding} task. 
{\colg T}he cost for an agent {\colg with regard to a} task is considered as a function of the distance from the agent to the task. 
The allocation results seem to be reasonable with consideration of the task rewards and the costs.

The number of iterations required to find a Nash stable partition is 1355, 1380, and 1295 for Scenario \#1, \#2, and \#3, respectively. 
The number of dummy iterations happened is {\colg just} $20$--$30\%$ of that of the iterations. 
{\colg This value is much} fewer than the results in Figure \ref{fig_result_scalability} because 
the networks {\colg considered here} are more connected than {\colg those in Section \ref{sec:result_scalability}}.

\begin{figure}
\centering
\subfloat[Scenario \#1]{\includegraphics[width=0.8\linewidth]{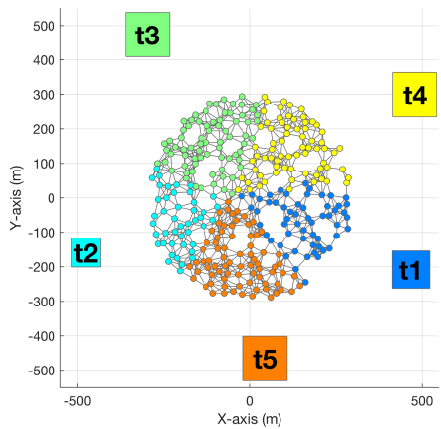}}
\hfil
\subfloat[Scenario \#2]{\includegraphics[width=0.8\linewidth]{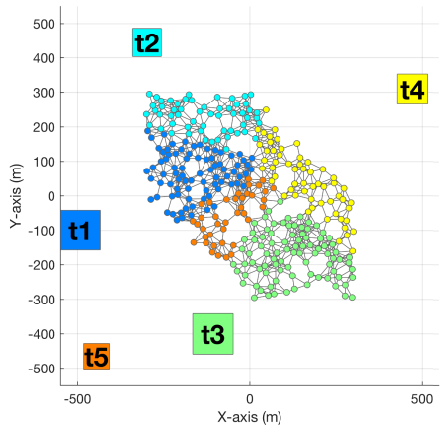}}
\hfil
\subfloat[Scenario \#3]{\includegraphics[width=0.8\linewidth]{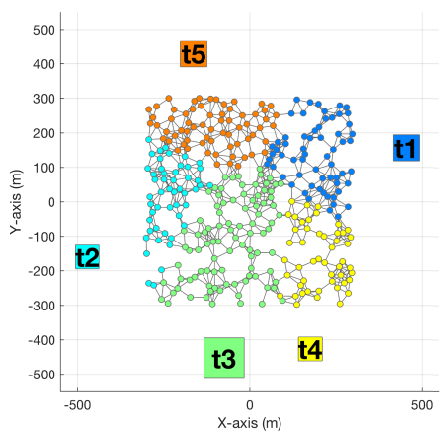}}
\hfil
\caption{{\col Visualize}d task allocation results {\col are shown with different geographic scenarios} ($n_t = 5$, $n_a = 320$). {\col Each square and its size represent each task's position and its reward (or demand), respectively. The circles and the lines between them indicate the positions of agents and their communication network, respectively. The color of each circle implies that the corresponding agent is assigned to the same colored task.}}
\label{fig_result_visual}
\end{figure}


\section{Conclusion}\label{Conclusion}

This paper proposed a novel game-theoretical framework that addresses a task allocation problem for a robotic swarm consisting of self-interested agents. 
We showed that selfish agents whose individual interests are transformable to SPAO preferences can converge to a Nash stable partition by using the proposed simple {\col decentralized} algorithm, which is executable even in asynchronous environments and under a strongly-connected communication network. 
We analytically and experimentally presented that the proposed framework provides scalability, a certain level of guaranteed suboptimality, adaptability, robustness, and a potential to accommodate different interests of agents.

As this framework can be considered as a new sub-branch of self-{\col organize}d approaches, 
one of our ongoing works is to compare it with one of the existing methods. 
Defining a fair scenario for both methods is non-trivial and requires careful consideration; otherwise, a resultant unsuitable scenario may provide biased results.
Secondly, another natural progression of this study is to relax anonymity of agents and thus to consider a combination of the agents' identities. 
Experimentally, we have often observed that heterogeneous agents with social inhibition also can converge to a Nash stable partition.  
More research would be needed to {\col analyze} the quality of a Nash stable partition obtained by the proposed framework in terms of $\min\max$ 
because our various experiments showed that the outcome provides individual utilities to agents in a balanced manner.

\section*{Acknowledgment}
The authors gratefully acknowledge that this research was supported by International Joint Research Programme with Chungnam National University (No. EFA3004Z)


\ifCLASSOPTIONcaptionsoff
  \newpage
\fi



\bibliographystyle{IEEEtran}
\bibliography{library}
\begin{IEEEbiographynophoto} \
 IEEE Copyright Notice $\copyright$ 2018 IEEE. Personal use of this material is permitted. Permission from IEEE must be obtained for all
other uses, in any current or future media, including reprinting/republishing this material for advertising
or promotional purposes, creating new collective works, for resale or redistribution to servers or lists, or
reuse of any copyrighted component of this work in other works.

\textbf{Accepted to be Published in: IEEE Transactions on Robotics}
\end{IEEEbiographynophoto}
%



\end{document}